\documentclass[aps,onecolumn,superscriptaddress,dvipsnames,nofootinbib,amssymb,floatfix,10pt,prb]{revtex4-2} 
\usepackage{graphicx}
\usepackage{dcolumn}
\usepackage{bm}
\usepackage{qcircuit}
\usepackage{braket}
\usepackage{color}
\usepackage{amsthm}
\usepackage{amsmath}
\usepackage{chemarrow}
\usepackage{overpic}
\usepackage{diagbox}
\usepackage{braket}
\usepackage{subfigure}
\usepackage{makecell}
\usepackage{mathtools} 
\usepackage{tcolorbox}
\newtheorem{theorem}{Theorem}
\newtheorem{definition}{Definition}
\newtheorem{lemma}{Lemma}
\newtheorem{corollary}{Corollary}
\usepackage{chngcntr}
\usepackage[colorlinks=true,linkcolor=red,bookmarks=true,breaklinks=true]{hyperref}
\let\oldacl\addcontentsline
\renewcommand{\addcontentsline}[3]{}
\newcommand{\Tr}{\mathrm{Tr}}

\DeclarePairedDelimiter\autobracket{(}{)}
\newcommand{\br}[1]{\autobracket*{#1}}
\DeclarePairedDelimiter\autoabsolute{|}{|}
\newcommand{\abs}[1]{\autoabsolute*{#1}}
\begin{document}
\title{
% Bracket entanglement: a resource indicator for classical simulation of quantum states\\
% Bracket entanglement governs the resources for preparation and simulation of quantum states\\
% Unveiling connections between tensor network and stabilizer formalism through bra-ket entanglement
% \\
%Bra-ket entanglement: an indicator that bridges classical simulation methods\\
Bra-ket entanglement, an indicator bridging entanglement, magic, and coherence}
\author{Zhong-Xia Shang}
\email{shangzx@hku.hk}
\affiliation{HK Institute of Quantum Science $\&$ Technology, The University of Hong Kong, Hong Kong}
\affiliation{QICI Quantum Information and Computation Initiative, School of Computing and Data Science,
The University of Hong Kong, Pokfulam Road, Hong Kong}

\author{Si-Yuan Chen}
\affiliation{Hefei National Research Center for Physical Sciences at the Microscale and School of Physical Sciences, University of Science and Technology of China, Hefei 230026, China}
\affiliation{Shanghai Research Center for Quantum Science and CAS Center for Excellence in Quantum Information and Quantum Physics, University of Science and Technology of China, Shanghai 201315, China}
\affiliation{Hefei National Laboratory, University of Science and Technology of China, Hefei 230088, China}

\author{Wenjun Yu}
\affiliation{QICI Quantum Information and Computation Initiative, School of Computing and Data Science,
The University of Hong Kong, Pokfulam Road, Hong Kong}

\author{Giulio Chiribella}
\email{giulio@cs.hku.hk}
\affiliation{QICI Quantum Information and Computation Initiative, School of Computing and Data Science,
The University of Hong Kong, Pokfulam Road, Hong Kong}

\author{Qi Zhao}
\email{zhaoqcs@hku.hk}
\affiliation{QICI Quantum Information and Computation Initiative, School of Computing and Data Science,
The University of Hong Kong, Pokfulam Road, Hong Kong}

\footnotetext{ZS and SC contributed equally to this work.}

\begin{abstract}
Understanding the intricate interplay between distinct quantum resources is a fundamental prerequisite for rigorously characterizing the boundary between classical and quantum technologies. Among the vast landscape of quantum resources, entanglement, magic, and coherence have arguably attracted the most intense investigation. However, while universally recognized as the core drivers of quantum advantage, our understanding of their structural interplay remains fragmented and compartmentalized. In this work, we introduce an indicator called {\em bra-ket entanglement} (BKE) defined in the operator vectorization space to bridge all three quantum resources. Specifically, we show that BKE governs a resource dependence transition in the generation of entanglement: in the low-BKE regime, the growth of entanglement is dominated by coherence, largely independent of magic. However, as BKE increases, the dependence on coherence will gradually be replaced by a dependence on magic. Consequently, in the high-BKE regime, entanglement generation becomes dominated by magic, largely independent of coherence. These results are built on a series of new entropy-theoretic relations and are verified through numerical experiments. We also discuss implications of our results for the resource transitions in classical simulations of mixed states and marginal probabilities and for relating different classical simulation methods.

% which can be used to lower bound the amount of coherence and magic resources needed to produce entanglement  by applying a given unitary to a given initial state.  In addition, bracket entanglement can be applied not only to quantum states but also observables. Finally, we show that it can also be used to choose whether the quantum state generated by a given unitary can be more easily simulated by the tensor network approach or by the stabilizer formalism.   

% Tensor network and stabilizer formalism are two common approaches for classical simulations of quantum systems. In this work, we explore their connections from a quantum resource perspective on simulating $UOU^\dag$ with $O$ an arbitrary operator ranging from pure state density matrices to Pauli operators. 

\end{abstract}
\maketitle
\noindent\textit{\textbf{Introduction.—}}The ongoing quest to achieve and definitively prove quantum advantages~\cite{preskill2018quantum,preskill2025beyond} relies fundamentally on our ability to understand and harness the underlying physical resources that remain inherently inaccessible to classical systems~\cite{chitambar2019quantum}. These unique quantum-mechanical resources serve a dual purpose in the landscape of modern quantum computing and information science \cite{nielsen2010quantum}. On one hand, they act as the essential fuel that powers emerging quantum technologies~\cite{bennett1993teleporting,ahnefeld2022role,zhao2025entanglement,kuroiwa2024every,zhang2024unconditional,yoganathan2019quantum,thomas2025role}. On the other hand, these very same resources establish the fundamental limits for classically simulating quantum systems~\cite{orus2019tensor,gottesman1998heisenberg,aaronson2004improved,vidal2003efficient,burgholzer2021hybrid}. Identifying and quantifying these resources is crucial for determining which quantum systems are "hard" to simulate and, consequently, where true quantum advantage lies.

Among the zoo of quantum features, three resources stand out as the pillars of quantum advantages: entanglement~\cite{horodecki2009quantum}, magic~\cite{veitch2014resource}, and coherence~\cite{streltsov2017colloquium}. {\em Entanglement}, an intrinsic quantum correlation that plays a key role in quantum information applications such as quantum teleportation~\cite{bennett1993teleporting} and quantum communication~\cite{ekert1991quantum}. In quantum computation, it also prevents efficient classical simulation by tensor network approaches~\cite{orus2019tensor,cirac2021matrix,perez2006matrix} and has recently been shown to accelerate quantum simulations \cite{zhao2025entanglement}. {\em Magic}, helps go beyond the stabilizer state, is a resource that enables universal quantum computation~\cite{bravyi2005universal,nielsen2010quantum} combined with Clifford circuits. Magic can be characterized as the level of superpositions of Pauli operators~\cite{leone2022stabilizer,dowling2024magic}, which set the bottleneck for efficient classical simulations via stabilizer-based simulation methods~\cite{gottesman1998heisenberg,aaronson2004improved} and Pauli-truncation based methods~\cite{schuster2024polynomial,aharonov2023polynomial,angrisani2025simulating,angrisani2024classically}. {\em Coherence}, capturing the ability of a quantum system to exist in a coherent superposition of reference basis, is critical for quantum speedup due to the quantum interference effect in quantum algorithms~\cite{stahlke2014quantum} and enables Heisenberg-limit quantum metrology protocols~\cite{giovannetti2011advances}. Coherence also sets the complexity lower bound for classical simulation methods based on the Feynman path~\cite{burgholzer2021hybrid,dawson2004quantum}.

Given the critical roles of these three resources, understanding their connections is a natural and necessary pursuit. Historically, however, they have been treated as distinct properties with different domains of application. For example, a stabilizer state \cite{hein2006entanglement} can possess maximal entanglement yet be efficiently simulatable due to zero magic \cite{gottesman1998heisenberg,aaronson2004improved}, while a product state with no entanglement can have both high magic \cite{haug2023quantifying} and high coherence \cite{streltsov2017colloquium}. The precise mechanism governing their interplay—and specifically, how one resource might be related to or constrain another—has remained unclear.

In this paper, we bridges these three resources by introducing an indicator which we call the {\em bra-ket entanglement} (BKE). Defined for general operators $O$ within the vectorized operator formalism, BKE is a conserved quantity under unitary evolution and can diagnose how the entanglement growth of the quantum evolution $UOU^\dag$ is governed by the interplay between the coherence and magic. Here, $O$ can represent either a quantum state or a quantum observable, and $U$ represents the quantum evolution in the Schr\"odinger or in the Heisenberg picture, respectively. 

Through BKE, we identify a resource dependence transition: when BKE of $O$ is low, entanglement growth is limited primarily by coherence (associated with Hadamard gates) and is independent of magic, explaining why stabilizer states can have high entanglement without magic. However, as BKE of $O$ increases, the dominant resource responsible for entanglement growth shifts from coherence to magic (associated with T gates), providing a deep underlying mechanism for the correlations between operator entanglement and Heisenberg magic observed in a recent study~\cite{dowling2025bridging}. Therefore, we highlight the importance of this quantity of BKE in understanding and relating different quantum resources.

\noindent\textit{\textbf{Spatial  vs bra-ket entanglement.—}}Throughout this work, we will adopt the vectorization picture to present our result, whose benifits will soon be clear.  Let $O$ be an operator acting on   the Hilbert space of  $n$ qubits, hereafter denoted by ${\cal H}_n  :  =  \mathbb C^{2^n}$,  and  let  $|O\rangle    :=  \sum_{i,j}  \, \langle i| O |j\rangle/\sqrt{\Tr[O^\dag O]} \,  |i\rangle  |j\rangle$ be the corresponding  normalized  vector  in the tensor product  Hilbert space ${\cal H}_n \otimes {\cal H}_n$.  Here, $O$ can be interpreted to represent either the density matrix or an observable.

\begin{figure}[htpb]
\centering
\includegraphics[width=0.7\linewidth]{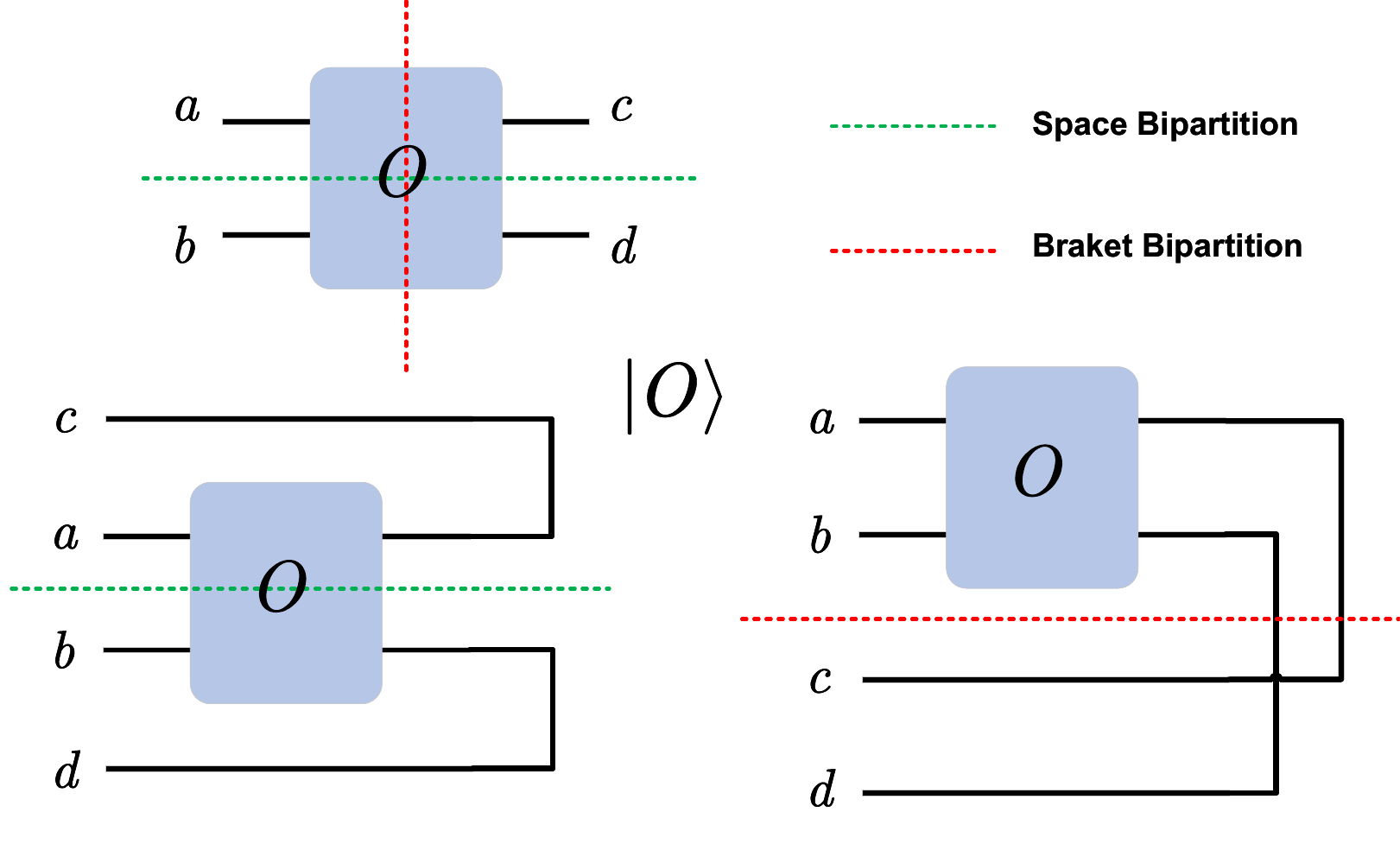}
\caption{\small{Tensor network representation of vectorization of $O$ and two types of bipartitions: space bipartition and bra-ket bipartition. Taking 2-qubit $O$ as an example. When $O=|00\rangle\langle 00|$, $|O\rangle$ is a product state under both space and bra-ket bipartitions. When $O$ is a Bell state density matrix, $|O\rangle$ is a maximally entangled state under the space bipartition but a product state under the bra-ket bipartition. When $O$ is proportional to a 2-qubit Pauli operator, $|O\rangle$ is a product state in the space bipartition but a maximally entangled state in the bra-ket bipartition.}
\label{fig1}}
\end{figure}  

When the $n$ qubits are divided into two groups which we call the space bipartition, the operator $O$ can be viewed as a tensor with four indices $a,b,c,d$, with indices $a$ and $c$ ($b$ and $d$) associated to the upper (lower) subsystem, as illustrated in Fig. \ref{fig1}. With respect to this bipartition, $|O\rangle$ can be written in the Schmidt decomposition $|O\rangle=\sum_i s_i |O_{u,i}\rangle\otimes |O_{l,i}\rangle$, where  \{$|O_{u,i}\rangle\}_i$  ($\{|O_{l,i}\rangle\}_i$)  are orthogonal vectors acting on the upper (lower) subsystems with subscripts ${\rm u}({\rm l})$. Since  $|O\rangle$ is normalized, we can regard it as a pure state of $2n$ qubits in ${\cal H}_n \otimes {\cal H}_n$, whose entanglement can be quantified as  
\begin{align}
H_{SE,\alpha}\br{O}:=    H_\alpha  \left( \, \{  s_i^2 \} \right) \, ,
\end{align}
where $H_\alpha  (\{p_i\})  :=  \log \left( \sum_i  p_i^\alpha \right)/(1-\alpha)$ is the $\alpha$-order  Rényi entropy~\cite{bromiley2004shannon}. 
 We call $ H_{SE,\alpha}\br{O}$ the {\em space entanglement} (SE) of the operator $O$.  When $O$ is the density matrix of a pure quantum state, its SE is exactly twice the entanglement entropy of the state \cite{horodecki2009quantum}. When the operator $O$ is the density matrix of a mixed state, the spatial entanglement is generally not a measure of mixed state entanglement (for example, it can be nonzero even for some separable states)~\cite{bennett1996mixed,horodecki1998mixed}, but it is important because it determines the bond dimension, which in turn determines the complexity of classical simulations in the tensor network framework \cite{orus2019tensor}.  In the literature,  $H_{SE,\alpha}\br{O}$  is sometimes known as the  {\em operator entanglement} of operator $O$ \cite{prosen2007operator}.

We now introduce another type of entanglement, defined by taking the bipartition between the Hilbert spaces associated with bras and kets in the operator $O$, which we call the bra-ket bipartition ($(ab|cd)$ in Fig. \ref{fig1}).  The indices $a$ and $b$  ($c$ and $d$) correspond to the rows (columns) of the operator $O$. Taking the Schmidt decomposition of the vector $|O\rangle$ with respect to this bipartition, we then obtain a different notion of entanglement, which we call the {\em bra-ket entanglement} (BKE) of operator $O$ and is quantified as  
 \begin{align}
    H_{BKE,\alpha}\br{O}:=   H_\alpha \left( \{  \gamma_i^{2}  \} \right)\,,
    \end{align}
where  $\{\gamma_i\}$ are the Schmidt coefficients in the decomposition  $|O\rangle=\sum_i \gamma_i |O_{r,i}\rangle\otimes |O_{c,i}\rangle$, where  \{$|O_{r,i}\rangle\}_i$  ($\{|O_{c,i}\rangle\}_i$)  are orthogonal vectors acting on the Hilbert spaces associated to the row (column) indices (See SM \ref{BKE}).  For example, when $O$ is a rank-one operator like a pure state density matrix, $|O\rangle$ is a product state with zero BKE. In contrast, when $O$ is an $n$-qubit Pauli operator,  $|O\rangle$ is proportional to a Bell state and has maximal BKE.  In later sections, we will see how BKE relates the SE with two other quantum resources, namely magic and coherence. An important and nice feature of BKE is that it is a conserved quantity under unitary transformations (See SM \ref{BKE}).

\noindent\textit{\textbf{Magic and coherence in the vectorization picture.—}}Magic quantifies the deviation from stabilizer states, and the cost of simulating non-stabilizer states in the stabilizer framework \cite{gottesman1998heisenberg,aaronson2004improved}.  Here, we define a quantifier of magic in the vectorization picture. Our approach is to expand  the vector $|O\rangle$ in the Bell basis with respect to the bipartition $(ab|cd)$  between bras and kets in the original operator $O$, and to define the {\em Bell basis coherence} (BBC) as the $\alpha$-order R\'enyi entropy     
\begin{align}
H_{BBC,\alpha}(O)  :   = H_\alpha  \left(   \{ |b_i|^2  \} \right) \, , 
\end{align}
where $\{b_i\}_i$ are the expansion coefficients in the decomposition  
$|O\rangle = \sum_i b_i |B_i\rangle$, where $\{ |B_i\rangle\}$ is the Bell basis.  When $O$ is the density matrix of a pure state,  $H_{BBC,\alpha}(O)$ coincides with the shifted {\em stabilizer entropy} \cite{leone2022stabilizer}, a common measure of magic for pure states. For general operators $O$, it coincides with the  {\em operator stabilizer Rényi entropy}~\cite{dowling2024magic}. (Also see SM \ref{bbcco}.)

A similar approach can be taken for coherence. We quantify the   \emph{computational basis coherence } (CBC) of an operator $O$ in terms of the $\alpha$-order R\'ernyi entropy 
\begin{align} H_{CBC,\alpha}(O) = H_\alpha(\{|o_{ij}|^2\}) \,, 
\end{align}
where $\{o_{ij}\}$ are the coefficients in the expansion  
$|O\rangle =\sum_{ij}o_{ij}|i\rangle|j\rangle$.  When $O$ is the density matrix of a pure state, the $ H_{CBC,\alpha}(O)$ coincides with the twice of the normal coherence definition with the computation basis as the reference basis~\cite{streltsov2017colloquium}. (Also see SM \ref{cbcco}.)

We can find the remarkable advantage of adopting the vectorization picture. While coherence and magic are fundamentally different resources in the operator space, we can give them a unified treatment in the vectorization picture: {\em both coherence and magic are coherence in the vectorization picture, and their difference is just a matter of the reference basis, the computational basis, or the Bell basis.}

\noindent\textit{\textbf{Magic and coherence bound entanglement.—}}After the above definitions, we can now formally present our results. We fist show how the SE growth of $UOU^\dag$ is bounded by magic and coherence. To do so, we define three free circuit families. The CBC-free (no coherence) circuit family $\mathcal{F}_{CBC}$ and BBC-free (no magic) circuit family $\mathcal{F}_{BBC}$ are defined as sets of unitaries $U$ satisfying 
$H_{CBC,\alpha}(O) = H_{CBC,\alpha}(UO U^\dagger)$ and 
$H_{BBC,\alpha}(O) = H_{BBC,\alpha}(UO U^\dagger)$, respectively, for arbitrary $|O\rangle$. Note that in the vectorization picture, we have $UOU^\dag\rightarrow U\otimes U^*|O\rangle$. $\mathcal{F}_{CBC}$ is generated by diagonal phase and permutation gates on the computational basis, such as $\{T,S,CX\}$ or $\{CCX,S,CX\}$, while $\mathcal{F}_{BBC}$ consists of Clifford circuits generated by $\{H,S,CX\}$. The zero-coherence circuit family $\mathcal{F}_{zero}$ is the intersection of $\mathcal{F}_{CBC}$ and $\mathcal{F}_{BBC}$ with circuits constructed from $\{S,CX\}$ as examples. We discuss similarities and differences between $\mathcal{F}_{CBC}$ and $\mathcal{F}_{BBC}$ in SM~\ref{mcd}.

The SE generation power under a circuit family $\mathcal{F}$ is characterized by
\begin{eqnarray} \label{eq:G_rate}
&&G_{\mathcal{F},\alpha}(O):= \sup_{U\in \mathcal{F}}H_{SE,\alpha}(UO U^\dagger) \!\!- \!\!H_{SE,\alpha}(O),
\end{eqnarray}
with $G_{CBC,\alpha}(O)$, $G_{BBC,\alpha}(O)$, and $G_{zero,\alpha}(O)$ defined for 
$\mathcal{F} = \mathcal{F}_{CBC}, \mathcal{F}_{BBC},$ and $\mathcal{F}_{zero}$, respectively.
It is evident that
\begin{eqnarray}
G_{zero,\alpha}(O) \leq \min\left\{G_{CBC,\alpha}(O), G_{BBC,\alpha}(O)\right\}.
\end{eqnarray}
The key question is the comparison between $G_{CBC,\alpha}(O)$ and $G_{BBC,\alpha}(O)$, understanding which will clarify which resource, CBC or BBC, is more effective for generating higher SE in $\ket{O}$ and provide deeper insights into their relations. To answer this question, we first establish a general upper bound 
on space entanglement in terms of coherence under arbitrary 
tensor-product bases.
\begin{lemma}\label{mainlemma}
For any set of orthogonal product basis 
$\{|\phi_j\rangle\}$ under the space bipartition in vectorization space, the decomposition 
$O=\sum_j c_j |\phi_j\rangle$ satisfies
\begin{eqnarray}
H_{\alpha}(\{|c_j|^2\})\geq H_{SE,\alpha}(O).
\end{eqnarray}
\end{lemma}
\noindent The proof is given in SM~\ref{the1proof} with majorization relations~\cite{Nielsen_2000}.

In this work, we focus on two canonical choices: 
the computational basis $\{|i\rangle\!\langle j|\}$ and 
the Pauli basis $\{P_j/\sqrt{2^n}\}$, whose coefficient entropies 
correspond to CBC (coherence) and BBC (magic), respectively. 
Applying Lemma~\ref{mainlemma} to these two bases immediately yields 
the following theorem.
\color{black}
\begin{theorem}\label{mainthe1}
$G_{zero,\alpha}(O)$ has the upper bound
\begin{eqnarray}\label{zeroub}
\min\left\{n,H_{CBC,\alpha}(O),H_{BBC,\alpha}(O)\right\}-H_{SE,\alpha}(O).
\end{eqnarray}
$G_{CBC,\alpha}(O)$ has the upper bound
\begin{eqnarray}\label{ubcbc}
\min\left\{n,H_{CBC,\alpha}(O)\right\}-H_{SE,\alpha}(O).
\end{eqnarray}
$G_{BBC,\alpha}(O)$ has the upper bound
\begin{eqnarray}\label{ubbbc}
\min\left\{H_{BBC,\alpha}(O),n\right\}-H_{SE,\alpha}(O).
\end{eqnarray}
\end{theorem}
\noindent Note that CBC and BBC are respectively sensitive to Hadamard and $T$ gates, which are responsible for coherence and magic generation, respectively. For a circuit $U$ with $N_H$ Hadamards and $N_T$ $T$ gates, we have (see SM \ref{cbcco} and \ref{bbcco}):
\begin{eqnarray}
H_{CBC,\alpha}(UO U^\dagger) &\leq& 2N_H + H_{CBC,0}(O), \label{nhr} \\
H_{BBC,\alpha}(UO U^\dagger) &\leq& N_T + H_{BBC,0}(O). \label{ntr}
\end{eqnarray}
Therefore, in the absence of additional coherence and magic resources ($H$ and $T$), $G_{zero,\alpha}(O)$ is governed by the existing coherence and magic in $|O\rangle$. 
To generate higher entanglement, additional coherence and magic resources are needed, especially the smaller one. 
Theorem \ref{mainthe1} explains why the stabilizer states can achieve maximal entanglement. 
Starting from $O=|0\rangle\langle 0|^{\otimes n}=\prod_{i=1}^n(I_n+Z_i)/2$, $H_{CBC,\alpha}(O)=0$ severely limits $G_{zero,\alpha}(O)$, therefore, $H$ in $\mathcal{F}_{BBC}$ are necessary to increase RHS of Theorem \ref{mainthe1}. On the other hand, we have $H_{BBC,\alpha}(O)=n$ already reaches the largest achievable SE, meaning no need to introduce additional magic resources. Therefore, we have $0=G_{zero,\alpha}\br{|0\rangle^{\otimes n}| 0\rangle^{\otimes n}}=G_{CBC,\alpha}\br{|0\rangle^{\otimes n}| 0\rangle^{\otimes n}}$ while $G_{BBC,\alpha}\br{|0\rangle^{\otimes n}| 0\rangle^{\otimes n}}=n$.
Similar reasoning applies to $O=P$, with opposite behaviors. 

% Eq.~\eqref{zeroub} shows that the maximum SE achievable by an $H$- and $T$-free circuit is upper bounded by the smaller of the input's CBC and BBC values. Eqs.~\eqref{ubcbc} and \eqref{ubbbc} indicate that adding $H$ and $T$ gates relaxes these bounds respectively, enabling larger SE generation.

\noindent\textit{\textbf{BKE decides magic and coherence.—}}After showing that the entanglement growth is bounded by the smaller one between the coherence and the magic, now the second step is to understand what determines coherence and magic. From above, we observed that for $O=|0\rangle\langle 0|^{\otimes n}$ with minimal BKE, coherence is required to generate higher entanglement, while for $O=P$, with maximal BKE, magic is necessary to generate higher entanglement. Therefore, the BKE of $|O\rangle$ appears to influence its coherence and magic. We now show this is indeed the case, but before going on, we need to introduce the {\em{Fourier bra-ket entanglement}} (Fourier-BKE). Fourier-BKE is defined through coefficients $\{\lambda_i\}$, which are obtained from the Walsh-Hadamard transform \cite{o2014analysis} to the Schmidt coefficients $\{\gamma_i\}$ of BKE. $\{\gamma_i\}$ satisfies $\lambda_i\geq 0$ and $\sum_i\lambda_i^2=1$. Two sets of coefficients are related through $\Sigma=\sum_i\gamma_i |i\rangle\langle i|=2^{-n/2}\sum_i \lambda_i Q_i$ with $Q_i\in\{\pm 1\}\times\{I,Z\}^{\otimes n}$ (See SM \ref{BKE}). The Rényi entropy on $\{\lambda_i\}$ is referred to as the Fourier-BKE Rényi entropy
\begin{align}
H_{FBKE,\alpha}(O):= H_{\alpha}\br{\{\lambda_i^{2}\}} .
\end{align}
For example, when $O=|i\rangle\langle j|$, $|O\rangle=|i\rangle|j\rangle$ is a product state under bra-ket bipartition with minimal BKE but maximal Fourier-BKE. 
In contrast, when $O=2^{-n/2}P_i$ is a Pauli basis, $|O\rangle=|B_i\rangle$ is a Bell state with maximal BKE but minimal Fourier-BKE.
Both BKE and Fourier-BKE are conserved values under unitary transformations.

After the introduction of both BKE and Fourier-BKE, we are now ready to show Theorem~\ref{mainthe2}.
\begin{theorem}\label{mainthe2}
The CBC and BBC of $|O\rangle$ have the lower bounds 
\begin{eqnarray}
H_{CBC,\alpha}\br{O}&&\geq n-H_{FBKE,1/2}\br{O},\\
H_{BBC,\alpha}\br{O}&&\geq n-H_{BKE,1/2}\br{O}.
\end{eqnarray}
\end{theorem}
\noindent Detailed proofs can be found in SM \ref{the2proof}. Theorem \ref{mainthe2} implies that BKE and Fourier-BKE set lower bounds for BBC and CBC. In this sense, coherence and magic are like conjugate quantities that share the same position in deciding entanglement growth. Since CBC and BBC involve projecting $|O\rangle$ onto two different bases and BKE and Fourier-BKE are related by the Walsh-Hadamard transform, each of them also satisfies an entropic uncertainty principle \cite{coles2017entropic}.
\begin{theorem}\label{mainthe3}
CBC-BBC uncertainty principle and BKE Fourier-BKE uncertainty principle say
\begin{eqnarray}
H_{CBC,\alpha}(O)+H_{BBC,\beta}(O)&&\geq n,\\
H_{BKE,\alpha}(O)+H_{FBKE,\beta}(O)&&\geq n, 
\end{eqnarray}
where $\alpha$ and $\beta$ satisfy $1/\alpha+1/\beta=2$.
\end{theorem}
\noindent Detailed proofs can be found in SM \ref{BKE} and \ref{the1proof}. Due to the uncertainty principles in Theorem \ref{mainthe3}, CBC and BBC can neither be simultaneously small, nor can their lower bounds be simultaneously large.

\noindent\textit{\textbf{Finishing the picture.—}}Combining Theorems~\ref{mainthe1}--\ref{mainthe3}, 
the relation between entanglement, magic, and coherence can now be summarized in Fig.~\ref{fig2}. 
When BKE is low, Theorem~\ref{mainthe2} forces BBC to be large while CBC can be small, so CBC becomes the bottleneck for entanglement growth and coherence ($H$ gates) is the dominant resource for increasing SE. 
Conversely, when BKE is high (Fourier-BKE low), CBC is forced large while BBC can be small, making magic ($T$ gates) the dominant resource. 

This picture becomes fully quantitative for operators that saturate the inequalities in Theorems~\ref{mainthe2} and~\ref{mainthe3}, which we call \emph{BKE-matching operators} (See formal definition in SM~\ref{the2proof}). 
For such operators with $H_{\mathrm{BKE},1/2}(O)=r$, the CBC and BBC values are exactly determined: $H_{\mathrm{CBC},\alpha}(O)=r$ and $H_{\mathrm{BBC},\alpha}(O)=n-r$ (see Lemma~\ref{ll7} in SM~\ref{the2proof}). Combined with the gate-counting bounds in Eqs.~\eqref{nhr}--\eqref{ntr}, this yields a precise gate-complexity characterization.

\begin{figure}[htpb]
\includegraphics[width=0.7\linewidth]{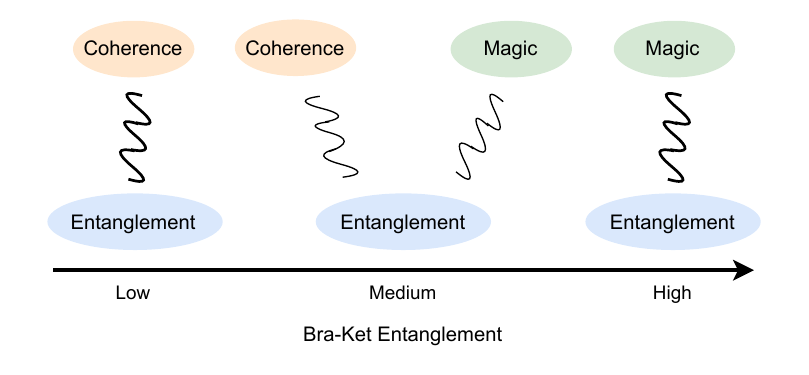}
\caption{\small{Relating (space) entanglement to CBC (coherence) and BBC (magic) via BKE. When BKE of $O$ is low, the SE is bounded by CBC, and therefore increasing SE requires introducing coherence resources. When BKE of $O$ is high, the SE is bounded by BBC, and therefore increasing SE requires introducing magic resources. When BKE is medium, both magic and coherence can relate SE, but not as strongly as in the extreme points.}
\label{fig2}}
\end{figure}

\begin{corollary}\label{lem=SE_bond}
For a BKE-matching operator $O$ with fixed bra-ket entanglement $H_{\mathrm{BKE},1/2}(O)=r$, if a circuit $U$ composed of $\{H,S,\mathrm{CNOT},T\}$ gates with $N_T$ number of $T$-gate and $N_H$ number of $H$-gate satisfies $H_{\mathrm{SE},\alpha}(U^\dagger O\, U)\geq\eta$, then
$N_T \geq \eta-(n-r), N_H \geq \eta-r$.
\end{corollary}

\noindent The proof follows from Lemma~\ref{ll7} in SM~\ref{the2proof}. 
When $r=0$ (e.g.\ $O=|0\rangle\!\langle 0|^{\otimes n}$), the $N_T$ bound becomes trivial while $N_H\geq\eta$ dominates, recovering the coherence-limited regime discussed after Theorem~\ref{mainthe1}: stabilizer states achieve maximal entanglement without magic. 
When $r=n$ (e.g.\ $O=P$), the situation reverses and $N_T\geq\eta$ dominates, explaining the entanglement--magic correlation observed in Heisenberg evolution of Pauli operators~\cite{dowling2025bridging}. 
At intermediate BKE, both bounds impose comparable constraints, so neither coherence nor magic alone capture the tensor-network complexity.

To verify the picture, in Fig.~\ref{fig3}(a) we numerically~\cite{xu2024mindspore,ChanceSiyuan_repo} show the SE of $U^\dagger O U$ with BKE-matching operators $O \in \{X^{\otimes 6},\ket{0}\!\bra{0}^{\otimes 6},X^{\otimes 3} \otimes\ket{0}\!\bra{0}^{\otimes 3}\}$. $U$ is generated randomly from the local gate set $\{CX,H,CCX,S\}$ with different $CCX$ rates and $H$ rates, indicating different levels of CBC and BBC resources. The reason of using $CCX$ as the magic gate instead of $T$ can be found in SM~\ref{whyccx}. Consistent with our theoretical predictions, for $X^{\otimes 6}$ with maximal BKE, SE is more sensitive to BBC resource, while for $|0\rangle\langle 0|^{\otimes 6}$ with minimal BKE, SE is more sensitive to CBC resource. For $X^{\otimes 3} \otimes\ket{0}\!\bra{0}^{\otimes 3}$, we have $H_{CBC,\alpha}(O)=H_{BBC,\alpha}(O)=3\leq 6$, making both CBC and BBC resources needed.

In Fig. \ref{fig3}(b), we quantify how the space-entanglement changes for BKE-matching operators $ O \in \{X^{\otimes 10} \otimes|0\rangle\!\langle 0|^{\otimes (10- s)}\}_{s = 0}^{10}$ in the $10$-qubits case, after the implementation of $U$ in three families $\mathcal{F}_{zero},\mathcal{F}_{BBC}$, and $\mathcal{F}_{CBC}$ respectively. For BKE-matching operators, we can find numerical results exactly match Theorem \ref{mainthe2}, where low BKE requires more CBC resources in $\mathcal{F}_{BBC}$ and high BKE needs more BBC resources in $\mathcal{F}_{CBC}$. Additionally, we observe that $U\in \mathcal{F}_{zero}$ can create a decent level of SE for $O$ in the middle, which is well captured by Theorem \ref{mainthe1} and is because these $O$ contain a decent amount of both BBC and CBC resources already.

\begin{figure}
\centering
\includegraphics[width=0.7\linewidth]{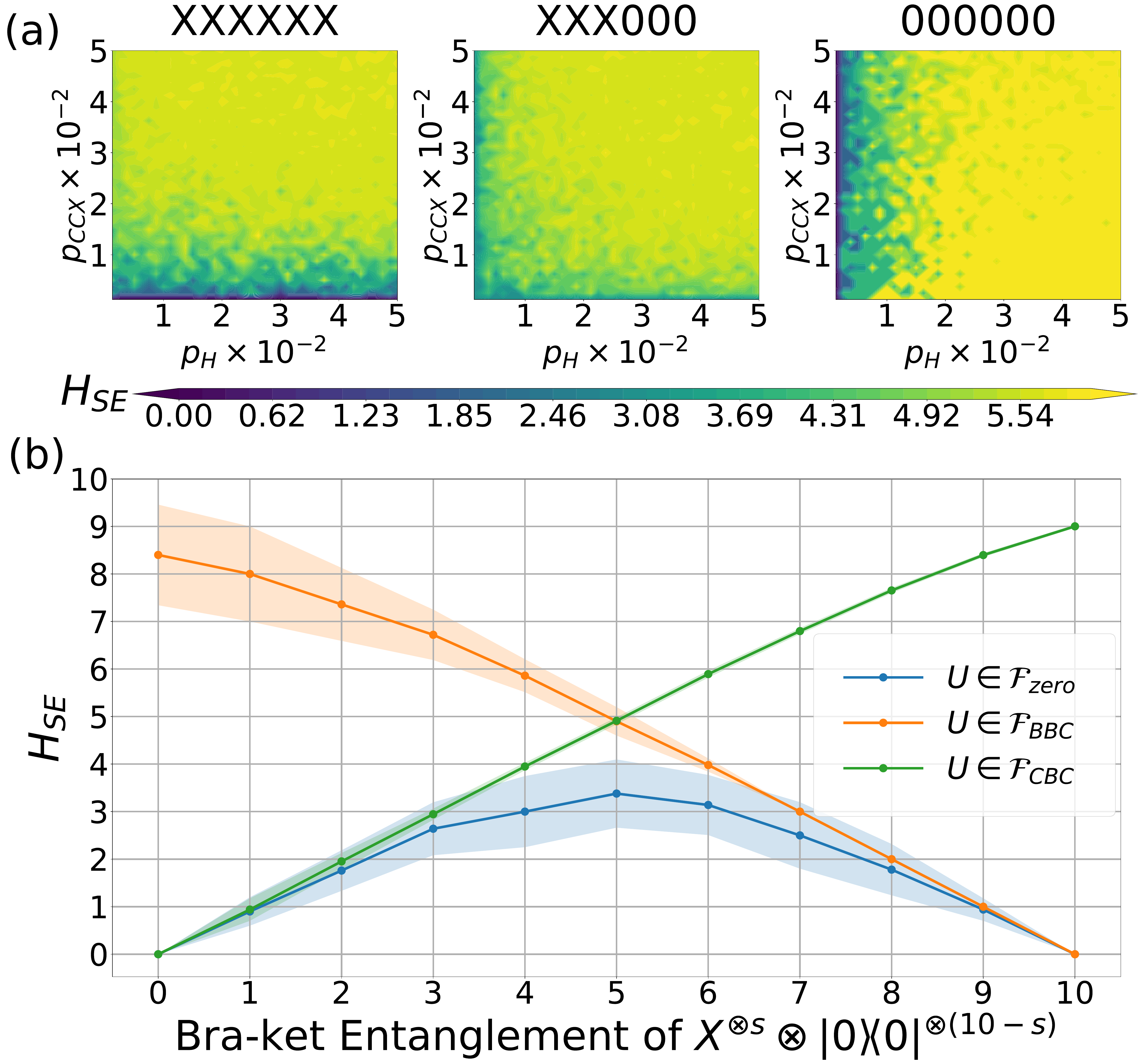}
\caption{\small{(a): Gradient diagram of SE of $UOU^\dag$ in different CBC and BBC resources. $U$ is composed of many layers (deep enough to approach allowable SE upper bounds). Each layer is randomly composed of $\{CX,H,CCX,S\}$ with $p_{CCX}$ and $p_H$ denotes the picking rate. (b): SE of $UOU^\dag$ as functions of BKE. $ O \in \{X^{\otimes s} \otimes|0\rangle\!\langle 0|^{\otimes (10- s)}\}_{s = 0}^{10}$ are BKE-matching operators with BKE ranging from 0 to 10. $U$ are deep circuits randomly chosen from $\mathcal{F}_{CBC}$ made up of $\{CCX,S,CX\}$ (SM \ref{whyccx}), $\mathcal{F}_{BBC}$ made up of $\{H,S,CX\}$, and $\mathcal{F}_{zero}$ made up of $\{S,CX\}$ respectively.
In both sub-figures, we choose $\alpha=1$ for $H_{SE}$.\label{fig3}}}
\end{figure}

\noindent\textit{\textbf{Implications for classical simulation.—}}We now show three implications of our results on classical simulations. The first two focus on the resource dependence transition for simulation methods related to entanglement, such as matrix product states and operators~\cite{perez2006matrix,cirac2021matrix}. The last one focus on the relations between different types of simulation methods.

{\em Resource transition in marginal probability simulation.} BKE-matching operators arise naturally in practical quantum simulation tasks (see Appendix~\ref{app:application} for details and other generalizations). 
As a key example, computing marginal probability distributions~\cite{bravyi2022simulate} over $n_2$ qubits of an $n=(n_1+n_2)$-qubit circuit $U$ amounts to evaluating 
\begin{align}
p(\mathbf{s})=\mathrm{tr}[|0\rangle\!\langle 0|^{\otimes n}
U^\dagger(\mathbb{I}_{n_1}\otimes|\mathbf{s}\rangle\!
\langle\mathbf{s}|U)].
\end{align}
The effective observable $O=\mathbb{I}_{n_1}\otimes|\mathbf{s}\rangle\!\langle\mathbf{s}|$ is a BKE-matching operator with $r=n_1$, so Corollary~\ref{lem=SE_bond} directly yields $N_T\geq\eta-n_2$ and $N_H\geq\eta-n_1$. 
This interpolates between the coherence-limited regime for amplitude estimation ($n_1=0$) and the magic-limited regime for Pauli average estimation ($n_1=n-1$)---two tasks that, even for the same $T$-depth-one circuit, exhibit a sharp complexity separation from GapP-complete to P~\cite{zhang2025classical}. 

% {\em Resource transition via purity or temperature.} It is not difficult to see that the purity of a density matrix tells its BKE at order 2. Therefore, we can consider the SE generation of $U\rho U^\dag$ for $\rho$ with different purities. For pure state, high purity implies low BKE, making coherence important for SE growth. When the purity of the initial state $\rho$ decreases, BKE will increase, therefore, magic will gradually takes the role of coherence for SE growth. This resource transition can also be applied for finite temperature quantum simulation. Since for Gibbs states, decreasing the temperature means increasing the purity, therefore, high-temperature quantum simulations will be more related to magic and less related to coherence compared to low-temperature quantum simulations.

{\em Resource dependence transition in mixed state simulation.} When $O=\rho$ is a density matrix, the order-2 BKE is then related to the purity of the state
\begin{eqnarray}
H_{\mathrm{BKE},2}(\rho) = -\log\mathrm{tr}(\rho^2).
\end{eqnarray}
Thus, using our estabilished picture above, we have the story for mixed state simulation. As purity decreases, BKE increases, shifting the dominant resource for SE growth from coherence ($N_H$) to magic ($N_T$).
This implies that high-temperature quantum simulations are more sensitive to magic, while low-temperature ones are more sensitive to coherence.

{\em Relating different classical simulation methods.} There are classical simulation methods that are directly related to entanglement, coherence, and magic, respectively. Examples are tensor-network-based methods (entanglement)~\cite{orus2019tensor,cirac2021matrix,perez2006matrix}, Stabilizer and Pauli truncation-based methods (magic)~\cite{gottesman1998heisenberg,aaronson2004improved,schuster2024polynomial,aharonov2023polynomial,angrisani2025simulating,angrisani2024classically}, and Feynman path-based methods (coherence)~\cite{burgholzer2021hybrid,dawson2004quantum}. Since Theorem \ref{mainthe1}, can be understood as BBC and CBC are lower bounded by SE, therefore, whenever a quantum system $UOU^\dag$ is difficult for entanglement-related methods, it must also be difficult for magic and coherence-related methods. Also, when BKE is low, entanglement-related methods can be more similar to coherence-related methods in terms of the algorithm complexity. In contrast, when BKE is high, entanglement-related methods will be more approaching to magic-related methods.

\noindent\textit{\textbf{Outlooks.---}}The desire to connect different concepts is rooted in the genes of physics. In this work, by introducing bra-ket entanglement as a diagnostic quantity, we have successfully bridged entanglement, coherence, and magic in a unified picture. We anticipate our results to have both significant impacts in quantum resource theories and in understanding the boundary between classical simulatabilities and quantum advantages.

Several open directions remain. Extending our framework to open quantum systems seems natural but challenging: quantum channels are merely non-unitary gates in the vectorization space, but can alter BKE. Also, although the computational basis and the Pauli basis can be viewed as gauge choices and can be unitarily transformed to other bases with the same BKE, they correspond to extreme minimal and maximal BKE values,
respectively. This raises the question of whether such extreme bases possess special positions, and whether results
such as Theorems 2 can be generalized to arbitrary matrix bases. Finally, since both CBC and BBC have non-zero lower bounds and therefore non-zero upper bounds for SE, no existing simulation methods have efficient and bounded complexity even for the circuit family $\mathcal{F}_{zero}$. Thus, we suggest for exploring specially-designed simulation methods for operators with intermediate BKE.

\textbf{Code: }We apply Mindspore Quantum~\cite{xu2024mindspore} for our numerical simulations with the source code~\cite{ChanceSiyuan_repo}.

\begin{acknowledgments}
The authors would like to thank Jakub Czartowski, Tianfeng Feng, Ray Ganardi, Tobias Haug, Chao-Yang Lu, Zhenhuan Liu, Nelly Ng, Jeongrak Son, Yunlong Xiao, Cheng-Cheng Yu, Wang Yao, Chengkai Zhu, and Xingjian Zhang for relevant discussions. This work was sponsored by CPS-Yangtze Delta Region IndustrialInnovation Center of Quantum and Information Technology-MindSpore Quantum Open Fund. Z.S. and Q.Z. acknowledge the support from Innovation Program for Quantum Science and Technology via Project 2024ZD0301900, National Natural Science Foundation of China (NSFC) via Project No. 12347104 and No. 12305030, Guangdong Basic and Applied Basic Research Foundation via Project 2023A1515012185, Hong Kong Research Grant Council (RGC) via No. 27300823, N\_HKU718/23, and R6010-23, Guangdong Provincial Quantum Science Strategic Initiative No. GDZX2303007, HKU Seed Fund for Basic Research for New Staff via Project 2201100596. Z.S. also acknowledges the support of HK Institute of Quantum Science and Technology.  G.C. acknowledges support from the Chinese Ministry of Science and Technology (MOST) through grant 2023ZD0300600, by the Hong Kong Research Grant Council (RGC) through grants  SRFS2021-7S02  and R7035-21F, and by the State Key Laboratory of Quantum Information Technologies and Materials. This work is also supported by the CPS-Yangtze Delta Region Industrial Innovation Center of Quantum and Information Technology-MindSpore Quantum Open Fund.

\end{acknowledgments}

\bibliography{ref}
\bibliographystyle{unsrt}
\clearpage

\begin{appendix}
\onecolumngrid
\renewcommand{\addcontentsline}{\oldacl}
\renewcommand{\tocname}{Supplementary material}
\tableofcontents
This supplementary material can be seen as a more detailed version compared to the main text.
\section{Preliminaries}
\subsection{Rényi entropy and entropic uncertainty principle}
\begin{definition}
Given a $D$-dimensional probability distribution $\{p_i\}$, the $\alpha$-order Rényi entropy \cite{bromiley2004shannon} with $0\leq\alpha\leq \infty$ and $\alpha\neq 1$ is defined as
\begin{eqnarray}
H_\alpha\br{\{p_i\}}=\frac{1}{1-\alpha}\log\br{\sum_i p_i^\alpha}.
\end{eqnarray}
\end{definition}
\noindent Throughout this work, we will use $\log\br{\cdot}$ with base 2. $H_\alpha\br{\{p_i\}}$ at $\alpha=1$ and $\alpha=\infty$ are defined by taking the limit, and we have
\begin{eqnarray}
H_1\br{\{p_i\}}=\lim_{\alpha\rightarrow 1}H_\alpha\br{\{p_i\}}&&=-\sum_i p_i\log\br{p_i}\text{, Shannon entropy,}\nonumber\\
H_\infty\br{\{p_i\}}=\lim_{\alpha\rightarrow \infty}H_\alpha\br{\{p_i\}}&&=-\log\br{\max_i p_i}.
\end{eqnarray}
Rényi entropy with different orders satisfy the relation
\begin{eqnarray}\label{renyiorder}
\log(D)=H_0(\{p_i\})\geq H_\alpha(\{p_i\})\geq H_\beta(\{p_i\})\geq H_\infty(\{p_i\}),\quad\forall~ 0\leq\alpha\leq\beta\leq\infty.
\end{eqnarray}

In quantum mechanics, measurements on different bases have the uncertainty principle, which can be well described through Rényi entropy. For any quantum state $\rho$, the relation between the entropy of the measurement probability distribution $\{p_{a,i}\}$ on the orthonormal basis $\{|i_a\rangle\langle i_a|\}$ and the entropy of the measurement probability distribution $\{p_{b,i}\}$ on the orthonormal basis $\{|i_b\rangle\langle i_b|\}$ satisfies the following lemma.
\begin{lemma}\label{entro}
The entropic uncertainty principle \cite{coles2017entropic} says
\begin{eqnarray}
H_{\alpha}\br{p_{a,i}}+H_{\beta}\br{p_{b,i}}\geq \log\br{\frac{1}{\max_{i,j} \abs{\langle i_a|j_b\rangle }^2}},\quad\forall~ \alpha,\beta\geq \frac{1}{2}\quad\text{and}\quad \frac{1}{\alpha}+\frac{1}{\beta}=2.
\end{eqnarray}
\end{lemma}

\subsection{Vectorization\label{bpcc}}

We will mainly use the vectorization picture for the presentation.
\begin{definition}[Vectorization]
The mapping vectorization
$\mathcal{V}[O]=\ket{O\rangle}$ maps a $n$-qubit matrix $O=\sum_{ij}o_{ij}|i\rangle\langle j|$ in operator space $\mathbb{C}^{2^n\times 2^n}$ to a $2n$-qubit vector $\ket{O\rangle}=\br{O\otimes I_n}\sum_j|j\rangle|j\rangle=\sum_{ij}o_{ij}|i\rangle|j\rangle$ in the vectorization space $\mathbb{C}^{2^{2n}}$. 
\end{definition}
\noindent $\ket{O\rangle}$ is also known as the Choi state of $O$ \cite{choi1975completely}. 

Here, we summarize the basic properties in the vectorization picture.
\begin{itemize}
\item \textbf{Norm and overlap.} The Frobenius norm of $O$ directly corresponds to the vector-2 norm of $\ket{O\rangle}$
\begin{eqnarray}
\|O\|_F^2=\Tr\br{O^\dag O}=\braket{\langle O|O\rangle}=\|\ket{O\rangle}\|_2^2.
\end{eqnarray}
In this work, we use the symbol $\ket{\cdot\rangle}$ for un-normalized vectors and $|\cdot\rangle$ for normalized vectors. Unless otherwise specified, in the following, without loss of generality, we will set $O$ to be normalized to facilitate the presentation.

The overlap between two operators $\Tr\br{O_1^\dag O_2}$ can be similarly translated
\begin{eqnarray}
\Tr\br{O_1^\dag O_2}=\braket{\langle O_1|O_2\rangle}.
\end{eqnarray}

\item \textbf{Complex conjugate.} For $O=|\phi\rangle\langle \psi|$ with $|\psi\rangle=\sum_i\psi_i|i\rangle$, we have $|O\rangle=|\phi\rangle|\psi^*\rangle$ with $|\psi^*\rangle=\sum_i\psi_i^*|i\rangle$.

\item \textbf{Operation.} A general linear transformation on $O$: $\sum_i A_i O B_i^\dag$ can be re-expressed in the vectorization picture
\begin{eqnarray}
\mathcal{V}\left[\sum_i A_i O B_i^\dag\right]=\sum_i A_i\otimes B_i^* \ket{O\rangle}.
\end{eqnarray}

\item \textbf{Computational basis.} Through vectorization, the matrix element basis will be mapped to the computational basis in the vectorization space
\begin{eqnarray}
|0\rangle\langle 0|&&\xrightarrow[]{\mathcal{V}}|0\rangle|0\rangle\nonumber,\\
|0\rangle\langle 1|&&\xrightarrow[]{\mathcal{V}}|0\rangle|1\rangle\nonumber,\\
|1\rangle\langle 0|&&\xrightarrow[]{\mathcal{V}}|1\rangle|0\rangle\nonumber,\\
|1\rangle\langle 1|&&\xrightarrow[]{\mathcal{V}}|1\rangle|1\rangle\nonumber.
\end{eqnarray}
\item \textbf{Pauli basis.} Through vectorization, the matrix Pauli basis will be mapped to the Bell basis in the vectorization space
\begin{eqnarray}\label{bpc}
\frac{1}{\sqrt{2}}I=\begin{pmatrix}1& 0\\0&1\end{pmatrix}&&\xrightarrow[]{\mathcal{V}}|\Phi^+\rangle=\frac{|0\rangle|0\rangle+|1\rangle|1\rangle}{\sqrt{2}}\nonumber,\\
\frac{1}{\sqrt{2}}X=\begin{pmatrix}0&1\\1&0\end{pmatrix}&&\xrightarrow[]{\mathcal{V}}|\Psi^+\rangle=\frac{|0\rangle|1\rangle+|1\rangle|0\rangle}{\sqrt{2}}\nonumber,\\
\frac{1}{\sqrt{2}}Y=\begin{pmatrix}0&-i\\i&0\end{pmatrix}&&\xrightarrow[]{\mathcal{V}}|\Psi^-\rangle=\frac{-i(|0\rangle|1\rangle-|1\rangle|0\rangle)}{\sqrt{2}}\nonumber,\\
\frac{1}{\sqrt{2}}Z=\begin{pmatrix}1& 0\\0&-1\end{pmatrix}&&\xrightarrow[]{\mathcal{V}}|\Phi^-\rangle=\frac{|0\rangle|0\rangle-|1\rangle|1\rangle}{\sqrt{2}}.
\end{eqnarray}
Note that, here, we change the phase convention of $|\Psi^-\rangle$ whose original definition is $\br{|0\rangle|1\rangle-|1\rangle|0\rangle}/\sqrt{2}$. In the following, we will mainly use $P$ to denote a $n$-qubit Pauli operator and use $|B\rangle$ to denote its corresponding vectorized $2n$-qubit Bell state.
\end{itemize}

\section{Entanglement in vectorization space}

As illustrated in Fig. \ref{fig1} in the main text, we will call the subsystem $ac$ the upper subsystem (US) and the subsystem $bd$ the lower subsystem (LS). And we will call the subsystem $ab$ the row subsystem (RS) and the subsystem $cd$ the column subsystem (CS).
\subsection{Space entanglement (SE)}
Space entanglement (SE) is related to space bipartition. In operator space, $O$ has the unique operator singular value decomposition
\begin{eqnarray}
O=\sum_i s_i O_{u,i}\otimes O_{l,i}.
\end{eqnarray}
Because $O$ is normalized, $\{s_i\}$ satisfy $s_i\geq 0$ and $\sum_i s_i^2=1$, and $O_{u,i}$ and $O_{l,i}$ are orthogonal and normalized matrices satisfying $\Tr\br{O_{u/l,i}^\dag O_{u/l,j}}=\delta_{ij}$. In vectorization picture, this operator Schmidt decomposition of $O$ corresponds to the Schmidt decomposition under $(ac|bd)$
\begin{eqnarray}
|O\rangle=\sum_i s_i |O_{u,i}\rangle\otimes |O_{l,i}\rangle,
\end{eqnarray}
where $|O_{u,i}\rangle$ and $|O_{l,i}\rangle$ are the vectorization of $O_{u,i}$ and $O_{l,i}$ in US and LS respectively. With the Schmidt coefficients $\{s_i\}$, we can define the SE-Rényi entropy 
\begin{definition}[SE-Rényi entropy]
The $\alpha$-order SE-Rényi entropy is defined as
\begin{eqnarray}
H_{SE,\alpha}(O)= \frac{1}{1-\alpha}\log\left(\sum_i s_i^{2\alpha}\right).
\end{eqnarray}
\end{definition}

The benefit of using SE in vectorization space is that it can be applied to arbitrary operators rather than just $O=|\psi\rangle\langle \psi|$.

\subsection{bra-ket entanglement (BKE)\label{BKE}}
bra-ket entanglement (BKE) is related to bra-ket bipartition. Under $(ab|cd)$, we can again write down a unique Schmidt decomposition of $|O\rangle$
\begin{eqnarray}\label{tesch}
|O\rangle=\sum_i \gamma_i |O_{r,i}\rangle\otimes |O_{c,i}\rangle,
\end{eqnarray}
where $\{\gamma_i\}$ satisfy $\gamma_i\geq 0$ and $\sum_i \gamma_i^2=1$, and $|O_{r,i}\rangle$ and $|O_{c,i}\rangle$ are orthogonal states in RS and CS respectively: $\langle O_{r,i} |O_{r,j}\rangle=\langle O_{c,i} |O_{c,j}\rangle=\delta_{ij}$. With the Schmidt coefficients $\{\gamma_i\}$, we can define the BKE-Rényi entropy. 
\begin{definition}[BKE-Rényi entropy]
The $\alpha$-order BKE-Rényi entropy is defined as
\begin{eqnarray}
\label{eq:te_def}
H_{BKE,\alpha}(O)= \frac{1}{1-\alpha}\log\left(\sum_i \gamma_i^{2\alpha}\right).
\end{eqnarray}
\end{definition}

For example, when $O=|i\rangle\langle j|$ is a matrix element basis, $|O\rangle=|i\rangle|j\rangle$ is a product state under bra-ket bipartition with $H_{BKE,\alpha}(|i\rangle|j\rangle)=0$. In contrast, when $O=2^{-n/2}P_i$ is a Pauli basis, $|O\rangle=|B_i\rangle$ is a Bell state with the maximal entanglement $H_{BKE,\alpha}(|B_i\rangle)=n$ under bra-ket bipartition.

In fact, the Schmidt decomposition Eq. \ref{tesch} under $(ab|cd)$ is a re-expression of the singular value decomposition of $O$,
\begin{eqnarray}\label{svd}
O=R\Sigma C,
\end{eqnarray}
where $R$ is a unitary operator whose columns are the orthonormal vectors $\{ |O_{r,i}\rangle \}$, $C$ is a unitary operator whose rows are given by the orthonormal vectors $\{ \langle O_{c,i}^*| \}$, and
\begin{eqnarray}\label{svd1}
\Sigma=\sum_i\gamma_i |i\rangle\langle i|.
\end{eqnarray}
Since $\Sigma$ is diagonal, we can also decompose $\Sigma$ into a linear combination of Pauli operators $Q_i\in\{\pm 1\}\times\{I,Z\}^{\otimes n}$ as
\begin{eqnarray}\label{svd2}
\Sigma=2^{-n/2}\sum_i \lambda_i Q_i,
\end{eqnarray}
with $\lambda_i\geq 0$ and $\sum_i\lambda_i^2=1$. For $\{\lambda_i\}$, we can again define its Rényi entropy, which we call the Fourier-BKE Rényi entropy.
\begin{definition}[Fourier-BKE Rényi entropy]
The $\alpha$-order Fourier-BKE Rényi entropy is defined as
\begin{eqnarray}
H_{FBKE,\alpha}(O)= \frac{1}{1-\alpha}\log\left(\sum_i \lambda_i^{2\alpha}\right).
\end{eqnarray}
\end{definition}

Since $\{\gamma_i\}$ and $\{\lambda_i\}$ are related by the Walsh-Hadamard transform \cite{o2014analysis}, we can prove a BKE Fourier-BKE uncertainty principle.
\begin{lemma}\label{TEFTE}[Theorem \ref{mainthe3} in main text]
The BKE Fourier-BKE uncertainty principle says
\begin{eqnarray}
H_{BKE,\alpha}(O)+H_{FBKE,\beta}(O)\geq n,\quad\forall~ \alpha,\beta\geq \frac{1}{2}\quad\text{and}\quad \frac{1}{\alpha}+\frac{1}{\beta}=2.
\end{eqnarray}
\end{lemma}
\begin{proof}
Since we have
\begin{eqnarray}
|\Sigma\rangle=\sum_i\gamma_i |i\rangle |i\rangle=\sum_i \lambda_i |Q_i\rangle,
\end{eqnarray}
and the maximal overlap between a 2n-qubit computational basis $|i\rangle |i\rangle$ and a Bell basis $|Q_i\rangle$ is $2^{-n/2}$, as shown in Eq. \ref{bpc}. Therefore, combined with Lemma \ref{entro}, Lemma \ref{TEFTE} is proved.
\end{proof}

An important property of BKE is that it remains unchanged under unitary transformation.
\begin{lemma}
For any unitary operator $U$, we have
\begin{eqnarray}
H_{BKE,\alpha}(O)&&=H_{BKE,\alpha}\left(UOU^\dagger\right),\\
H_{FBKE,\alpha}(O)&&=H_{FBKE,\alpha}\left(UOU^\dagger\right)
\end{eqnarray}
\end{lemma}
\begin{proof}
Since we have
$$\mathcal{V}\left[UOU^\dag\right]=U\otimes U^*|O\rangle,$$
with $U$ lives in RS and $U^*$ lives in CS, therefore, $U\otimes U^*$ is a product operator in bra-ket bipartition and won't change $\{\gamma_i\}$ and $\{\lambda_i\}$.
\end{proof}

\section{Coherence in vectorization space}
Quantum coherence is a basis-dependent concept. Here, we consider coherence in vectorization space under two bases: computational basis and Bell basis.

\subsection{Computational basis coherence (CBC)\label{cbcco}}
The operator $O$ can be decomposed in matrix element basis and can be further mapped into the computational basis in vectorization space
\begin{eqnarray}
O=\sum_{ij}o_{ij}|i\rangle\langle j|\xrightarrow[]{\mathcal{V}}|O\rangle=\sum_{ij}o_{ij}|i\rangle|j\rangle,
\end{eqnarray}
which leads to the definition of CBC-Rényi entropy.
\begin{definition}[CBC-Rényi entropy]
The $\alpha$-order CBC-Rényi entropy is defined as
\begin{eqnarray}
H_{\mathrm{CBC},\alpha}(O)= \frac{1}{1-\alpha}\log\left(\sum_{ij} |o_{ij}|^{2\alpha}\right).
\end{eqnarray}
\end{definition}
\noindent For example, when $O=|i\rangle\langle j|$ is a matrix element basis, we have $H_{\mathrm{CBC},\alpha}(|i\rangle\!\langle j|)=0$. In contrast, when $O=2^{-n/2}P_i$ is a Pauli basis, we have $H_{\mathrm{CBC},\alpha}(2^{-n/2}P_i) = (1-\alpha)^{-1}\log\left[2^{-\alpha n}\cdot 2^n\right]=n$.

Since CBC characterizes the distribution on the computational basis in vectorization space, when $U$ is constructed from the standard gate set $\{H,S,T,CX\}$, the Hadamard gate is the only gate that can change coherence since the other gates are only diagonal phase gates or permutation gates for computational basis. The effect of Hadamard gates on CBC can be summarized in the following Lemma.
\begin{lemma}\label{cbch}
Let $N_H$ be the number of Hadamard gates $H$ in $U$, we have
\begin{eqnarray}
H_{\mathrm{CBC},\alpha}(U^\dagger O U)\leq \min\{2N_H+H_{CBC,0}(O),2n\}.
\end{eqnarray}
\end{lemma}
\begin{proof}
According to Eq. \ref{renyiorder}, we have 
\begin{eqnarray}
H_{\mathrm{CBC},\alpha}(U^\dagger O U)\leq H_{CBC,0}(U^\dagger O U).
\end{eqnarray}
$H_{CBC,0}$ denotes the CBC rank: number of computational basis $|i\rangle|j\rangle$ with non-zero amplitudes for operators. The action of a single Hadamard gate can split $|i\rangle|j\rangle$ into a linear combination of 4 $|i\rangle|j\rangle$, making the number of computational basis $|i\rangle|j\rangle$ with non-zero amplitudes increase by at most four times. Therefore, for $N_H$ Hadamard gates, we must have
\begin{eqnarray}
H_{CBC,0}(U^\dagger O U)
&&\leq N_H \times \log 4+H_{CBC,0}(O)\\
&&= 2N_H+H_{CBC,0}(O).
\end{eqnarray}
Also, since there are $2^{2n}$ computational basis in vectorization space, another upper bound of $H_{CBC,0}(U^\dagger O U)$ is $2n$. 
\end{proof}

\subsection{Bell basis coherence (BBC)\label{bbcco}}
The operator $O$ can be decomposed in another basis: the Pauli basis, and can be further mapped into the Bell basis in vectorization space
\begin{eqnarray}\label{eq:def_BBdecomp}
O=2^{-n/2}\sum_{i}b_iP_i\xrightarrow[]{\mathcal{V}}|O\rangle=\sum_{i}b_i|B_i\rangle,
\end{eqnarray}
which leads to the definition of BBC-Rényi entropy.
\begin{definition}[BBC-Rényi entropy]
The $\alpha$-order BBC-Rényi entropy is defined as
\begin{eqnarray}
H_{\mathrm{BBC},\alpha}(O)= \frac{1}{1-\alpha}\log\left(\sum_i |b_i|^{2\alpha}\right).
\end{eqnarray}
\end{definition}
\noindent For example, when $O=|i\rangle\langle j|$ is a matrix element basis, we have $H_{\mathrm{BBC},\alpha}(|i\rangle|j\rangle)=n$. In contrast, when $O=2^{-n/2}P_i$ is a Pauli basis, we have $H_{\mathrm{BBC},\alpha}(|B_i\rangle)=0$.

BBC has strong connections with magic. We can consider the changing of BBC between $|O\rangle$ and $U\otimes U^*|O\rangle$. Since BBC characterizes the distribution of $O$ on the Pauli basis, when $U$ is constructed from the standard gate set $\{H,S,T,CX\}$, the $T$ gate is the only gate that can change $H_{\mathrm{BBC},\alpha}(O)$ since the other gates are all Clifford which are permutation gates on Pauli basis. The effect of $T$ on BBC can be summarized in the following lemma.
\begin{lemma}\label{bbct}
Let $N_T$ be the number of $T$ gates in $U$, we have
\begin{eqnarray}
H_{\mathrm{BBC},\alpha}(U^\dagger O U)\leq \min\{N_T+H_{BBC,0}(O), 2n\}.
\end{eqnarray}
\end{lemma}
\begin{proof}
According to Eq. \ref{renyiorder}, we have 
\begin{eqnarray}
H_{\mathrm{BBC},\alpha}(U^\dagger O U)\leq H_{BBC,0}(U^\dagger O U).
\end{eqnarray}
$H_{BBC,0}$ denotes the BBC rank: number of Bell basis $|B_i\rangle$ with non-zero amplitudes for operators. The action of a single $T$ gate can at most split $|B_i\rangle$ into a linear combination of 2 $|B_i\rangle$, making the number of Bell basis $|B_i\rangle$ with non-zero amplitudes increase by at most two times. Therefore, for $N_T$ the number of $T$ gates, we must have
\begin{eqnarray}
H_{BBC,0}(U^\dagger O U)\leq N_T+H_{BBC,0}(O).
\end{eqnarray}
Also, since there are $2^{2n}$ Bell basis in vectorization space, another upper bound of $H_{BBC,0}(U^\dagger O U)$ is $2n$. 
\end{proof}

\section{BKE diagnoses SE}

\subsection{CBC and BBC bound SE\label{the1proof}}
Through the above definitions, we can now give a formal theorem to relate SE with CBC and BBC.
\begin{lemma}\label{lemma:basis_upper_bound}
    Given a $n$ qubit system with basis $\{B_1,\cdots,B_{n^2}\}$ such that $\text{tr}(B_j B_k^\dagger) = \delta_{jk}$ and $B_j = O_j^{(1)}\otimes \cdots\otimes O_j^{(n)}$ are tensor product operators. Then the decomposition $O = \sum_{j=1}^{n^2}c_j B_j$ satisfied $H_{\alpha}(\{c_i^2\}) \geq H_{SE, \alpha}(O)$
\end{lemma}
\begin{proof}
    Rényi entropy are Schur concave functions~\cite{ash2012information}, therefore, it is sufficient to show the majorization relations$\{s_i^2\}\succeq\{c_i^2\}$.
The vectorization of basis $B_j$ forms an orthogonal basis set $\{|B_j\rangle\}$.
Because $B_j$ are tensor product operators and normalized, it is always possible to decompose its vectorized state as $|B_{\alpha,\beta}\rangle = |\alpha\rangle_u|\beta\rangle_l$ on the basis set $\{|\alpha\rangle_u|\beta\rangle_l\}$, where  $\{|\alpha\rangle_u\}$ and $\{|\beta\rangle_l\}$ are orthogonal basis sets in US and LS. 

Then we can decompose $|O\rangle$ on the basis set $\{|\alpha\rangle_u|\beta\rangle_l\}$
\begin{eqnarray}
|O\rangle=\sum_{\alpha\beta}c_{\alpha\beta}|\alpha\rangle_u|\beta\rangle_l,
\end{eqnarray}
with $\sum_{\alpha\beta} |c_{\alpha\beta}|^2=1$. We can re-express $|O\rangle$ as
\begin{eqnarray}
|O\rangle=\sum_{\alpha} \sqrt{\sum_{\beta_1} |c_{\alpha\beta_1}|^2}|\alpha\rangle_u\left(\frac{\sum_\beta c_{\alpha\beta}|\beta\rangle_l}{\sqrt{\sum_{\beta_2} |c_{\alpha\beta_2}|^2}}\right),
\end{eqnarray}
with $\sum_\alpha\left(\sqrt{\sum_{\beta} |c_{\alpha\beta}|^2}\right)^2=1$. It is obvious to see that
\begin{eqnarray}
\left\{\sum_{\beta} |c_{\alpha\beta}|^2\right\}\succeq \{|c_{\alpha\beta}|^2\}.
\end{eqnarray}
Since we also have [Corollary 4 in Ref \cite{Nielsen_2000}]
\begin{eqnarray}
\{s_i^2\}\succeq \left\{\sum_{\beta} |c_{\alpha\beta}|^2\right\},
\end{eqnarray}
therefore,
\begin{eqnarray}\label{major}
H_{\alpha}\left(\{|c_{\alpha\beta}|^2\}\right)\geq H_{SE,\alpha}(O).
\end{eqnarray}
\end{proof}
Based on Lemma~\ref{lemma:basis_upper_bound}, we prove the theorem below.
\begin{theorem}\label{the1}[Theorem \ref{mainthe1} in main text]
$G_{zero,\alpha}(O)$ has the upper bound
\begin{eqnarray}
\label{eq:zero_eq}
G_{zero,\alpha}(O)\leq\min\left\{n,H_{\mathrm{CBC},\alpha}(O),H_{\mathrm{BBC},\alpha}(O)\right\}-H_{SE,\alpha}(O).
\end{eqnarray}
$G_{\mathrm{CBC},\alpha}(O)$ has the upper bound
\begin{eqnarray}
\label{eq:CBC_eq}
G_{\mathrm{CBC},\alpha}(O)\leq\min\left\{n,H_{\mathrm{CBC},\alpha}(O)\right\}-H_{SE,\alpha}(O).
\end{eqnarray}
$G_{\mathrm{BBC},\alpha}(O)$ has the upper bound
\begin{eqnarray}
\label{eq:BBC_eq}
G_{\mathrm{BBC},\alpha}(O)\leq \min\left\{H_{\mathrm{BBC},\alpha}(O),n\right\}-H_{SE,\alpha}(O).
\end{eqnarray}

\end{theorem}
\begin{proof}
Because Eq. \ref{eq:CBC_eq}-\ref{eq:BBC_eq} are just corollaries of Eq. \ref{eq:zero_eq},  it is sufficient to prove the first equation.
By the definition of $G_{\mathcal{F}, \alpha}(O)$, it is sufficient to show that $\forall U \in \mathcal{F}_{\mathrm{zero}}$,
\begin{align}
 H_{SE,\alpha}(UO U^\dagger ) &\leq  n\\
    H_{SE,\alpha}(UO U^\dagger ) &\leq  H_{\mathrm{CBC},\alpha}(O)\\
       H_{SE,\alpha}(UO U^\dagger ) &\leq  H_{\mathrm{BBC},\alpha}(O)
\end{align}
Since we have
\begin{eqnarray}
H_{SE,\alpha}(UO U^\dagger )\leq H_{SE,0}(UO U^\dagger )\leq n,
\end{eqnarray}
the first identity is proved. By the definition of $\mathcal{F}_{\mathrm{zero}}$, we note that
\begin{align}
    H_{\mathrm{CBC},\alpha}(O) =  H_{\mathrm{CBC},\alpha}(UO U^\dagger ),\\
    H_{\mathrm{BBC},\alpha}(O) =  H_{\mathrm{BBC},\alpha}(UO U^\dagger ).
\end{align}
Thus, it is sufficient to prove that $\forall \ket{O}$,
\begin{eqnarray}
&&H_{\mathrm{CBC},\alpha}(O)\geq H_{SE,\alpha}(O),\label{eq:UB_SE1}\\
&&H_{\mathrm{BBC},\alpha}(O)\geq H_{SE,\alpha}(O).\label{eq:UB_SE2}
\end{eqnarray}
Since both the computational basis $\ket{i}\!\bra{j}$ and Pauli basis $P_{j}$ are both orthonormal tensor product bases, based on Lemma~\ref{lemma:basis_upper_bound}, we can safely state that 
\begin{eqnarray}
H_{\mathrm{CBC},\alpha}(O) &= H_{\alpha}\left(\{|c_{\alpha\beta}|^2\}\right)\geq H_{SE,\alpha}(O),\\
H_{\mathrm{BBC},\alpha}(O)&= H_{\alpha}\left(\{|b_{j}|^2\}\right)\geq H_{SE,\alpha}(O).
\end{eqnarray}
\end{proof}
Therefore, when no additional CBC and BBC resources ($H$ and $T$) are introduced, $G_{zero,\alpha}(O)$ is governed by the existing CBC and BBC in $|O\rangle$. To generate higher entanglement, additional CBC and BBC resources are needed. Since CBC and BBC are about projecting $|O\rangle$ on two different bases, following Lemma \ref{entro}, we can summarize a CBC-BBC uncertainty principle.
\begin{theorem}\label{the2}[Theorem \ref{mainthe3} in main text]
The CBC-BBC uncertainty principle says
\begin{eqnarray}
H_{\mathrm{CBC},\alpha}(O)+H_{BBC,\beta}(O)\geq n,\quad\forall~ \alpha,\beta\geq \frac{1}{2}\quad\text{and}\quad \frac{1}{\alpha}+\frac{1}{\beta}=2.
\end{eqnarray}
\end{theorem}
\begin{proof}
The maximal overlap between a 2-qubit computational basis and a Bell basis is $1/\sqrt{2}$, as shown in Eq. \ref{bpc}. Generalizing to $2n$-qubit cases, we have 
\begin{eqnarray}
\max_{i,j,k}\abs{\langle i|\langle j| B_k\rangle}^2=2^{-n},
\end{eqnarray}
which, combined with Lemma \ref{entro}, leads to the CBC-BBC uncertainty principle.
\end{proof}

\subsection{BKE bounds CBC and BBC\label{the2proof}}
We have seen from above that for $O=|0\rangle\langle 0|^{\otimes n}$, CBC is required for generating higher entanglement, while for $O=P$, BBC is required for generating higher entanglement. Therefore, a natural question is what properties of $|O\rangle$ decide which resource is more desired regarding increasing SE? 

An observation is that for $|O\rangle=|0\rangle^{\otimes n}| 0\rangle^{\otimes n}$, it is a product state in bra-ket bipartition $(ab|cd)$, while for $|O\rangle=|B\rangle$, it is a maximally entangled state in bra-ket bipartition $(ab|cd)$. Therefore, BKE of $|O\rangle$ seems to influence its CBC and BBC. Indeed, we prove the following main theorem.
\begin{theorem}\label{the3}[Theorem \ref{mainthe2} in main text]
The CBC and BBC of $|O\rangle$ have the lower bounds
\begin{eqnarray}
H_{\mathrm{CBC},\alpha}(O)&&\geq n-H_{FBKE,1/2}(O),\nonumber\\
H_{\mathrm{BBC},\alpha}(O)&&\geq n-H_{BKE,1/2}(O).
\end{eqnarray}
\end{theorem}
\begin{proof}
We first prove the CBC lower bound. From Eq. \ref{svd}, Eq. \ref{svd1}, and Eq. \ref{svd2}, we have the observation
\begin{eqnarray}\label{obser}
\frac{\abs{\langle i|O|j\rangle}}{2^{-n/2}\left(\sum_l\lambda_l\right)}=
\frac{\left|\sum_k\lambda_k\langle i|R Q_k C|j\rangle\right|}{\sum_l\lambda_l}
\leq\frac{\sum_k\lambda_k \cdot \left|\langle i|R Q_k C|j\rangle\right|}{\sum_l\lambda_l} \leq 1, 
\end{eqnarray}
where in the last inequality, we use the fact that $|\langle i|R Q_k C|j\rangle|\leq 1$ since $\langle i|R Q_k C|j\rangle$ is a quantum amplitude. Therefore, we have
\begin{eqnarray}
|o_{ij}|^{2\alpha}&&\geq 2^{n(1-\alpha)}\left(\sum_l\lambda_l\right)^{2\alpha-2}|o_{ij}|^{2},\quad\forall~ 0\leq\alpha<1,\\
|o_{ij}|^{2\alpha}&&\leq 2^{n(1-\alpha)}\left(\sum_l\lambda_l\right)^{2\alpha-2}|o_{ij}|^{2},\quad\forall~ \alpha> 1,
\end{eqnarray}
through which, for both $0\leq\alpha<1$ and $\alpha> 1$, we can bound $H_{\mathrm{CBC},\alpha}(O)$
\begin{eqnarray}
H_{\mathrm{CBC},\alpha}(O)&&= \frac{1}{1-\alpha}\log\left(\sum_{ij} |o_{ij}|^{2\alpha}\right)\\
&&\geq\frac{1}{1-\alpha}\log\left(\sum_{ij} 2^{n(1-\alpha)}\left(\sum_l\lambda_l\right)^{2\alpha-2}|o_{ij}|^{2}\right)\nonumber\\&&=n-H_{FBKE,1/2}(O)+\frac{1}{1-\alpha}\log\left(\sum_{ij}|o_{ij}|^2\right)\nonumber\\&&=n-H_{FBKE,1/2}(O),
\end{eqnarray}
where in the second line from the last, we use the definition $H_{FBKE,1/2}(O)=H_{1/2}\left(\{\lambda_i^2\}\right) = 2\log\left(\sum_i \lambda_i\right)$ and in the last line, we use the property $\sum_{ij}|o_{ij}|^2=1$. By taking the limit $\alpha\rightarrow1$, Rényi entropy becomes Shannon entropy, and the same result holds.

We now show the proofs for BBC. An observation similar to Eq. \ref{obser} is that
\begin{eqnarray}
\left|\frac{\Tr(OP_i)}{\sum_i\gamma_i}\right|=\left|\Tr\left(\frac{\Sigma}{\sum_i\gamma_i}CP_i R\right)\right|\leq 1
\end{eqnarray}
where the inequality is due to the fact that $\Sigma/\sum_i\gamma_i$ is a legal density matrix. Based on this observation, we have the relations
\begin{eqnarray}
\abs{\Tr(OP_i)}^{2\alpha}&&\geq\left(\sum_j\gamma_j\right)^{2\alpha-2}\left|\Tr(OP_i)\right|^{2}, \quad \forall~ 0\leq\alpha<1,\label{eq:inequlity1}\\
\abs{\Tr(OP_i)}^{2\alpha}&&\leq\left(\sum_j\gamma_j\right)^{2\alpha-2}\left|\Tr(OP_i)\right|^{2},\quad\forall~\alpha> 1.\label{eq:inequlity12}
\end{eqnarray}
By definition~\eqref{eq:def_BBdecomp}, we have $c_i = 2^{-n/2}\text{Tr}(P_iO)$. 
Then for both $0\leq\alpha<1$ and $\alpha> 1$,due to the different signal of $(1 - \alpha)^{-1}$,  $H_{\mathrm{BBC},\alpha}(O)$ can always be lower bounded as
\begin{eqnarray}
H_{\mathrm{BBC},\alpha}(O)&&=\frac{1}{1-\alpha}\log\left(\sum_i |b_i|^{2\alpha}\right)=\frac{1}{1-\alpha}\log\left(2^{-n\alpha}\sum_i |\Tr(P_i O)|^{2\alpha}\right)\nonumber\\&&\geq \frac{1}{1-\alpha}\log\left(2^{-n\alpha}\sum_i \left(\sum_j\gamma_j\right)^{2\alpha-2}|\Tr(OP_i)|^{2}\right)\nonumber\\&&=\frac{1}{1-\alpha}\log\left(2^{-n\alpha} \left(\sum_j\gamma_j\right)^{2\alpha-2}\right)+\frac{1}{1-\alpha}\log\left(\sum_i|\Tr(OP_i)|^{2}\right)\nonumber\\&&=\frac{-n\alpha}{1-\alpha}-H_{BKE,1/2}(O)+\frac{n}{1-\alpha}\nonumber\\&&=n-H_{BKE,1/2}(O),
\end{eqnarray}
where in the second line from the last, we use the definition $H_{BKE,1/2}(O)=H_{1/2}\left(\{\gamma_i^2\}\right) = 2\log\left(\sum_i \gamma_i\right)$ and the property
\begin{eqnarray}
\sum_i|\Tr(OP_i)|^{2}=2^n \text{Tr}\br{OO^\dagger}=2^n.
\end{eqnarray}
By taking the limit $\alpha\rightarrow1$, Rényi entropy becomes Shannon entropy, and the same result holds. 
\end{proof}

In the following box, we summarize the results from Theorem \ref{the2}-\ref{the3} and Lemma \ref{TEFTE}.
\begin{tcolorbox}
\textbf{Summary of relations around CBC, BBC, and BKE.}
$$H_{\mathrm{CBC},\alpha}(O)+H_{BBC,\beta}(O)\geq n,\quad\forall~ \alpha,\beta\geq \frac{1}{2}\quad\text{and}\quad \frac{1}{\alpha}+\frac{1}{\beta}=2,$$
$$H_{\mathrm{CBC},\alpha}(O)\geq n-H_{FBKE,1/2}(O),$$
$$H_{\mathrm{BBC},\alpha}(O)\geq n-H_{BKE,1/2}(O),$$
$$H_{BKE,\alpha}(O)+H_{FBKE,\beta}(O)\geq n,\quad\forall~ \alpha,\beta\geq \frac{1}{2}\quad\text{and}\quad \frac{1}{\alpha}+\frac{1}{\beta}=2.$$
\end{tcolorbox}
\noindent Through these inequalities, we can define the BKE-matching operator.
\begin{definition}[BKE-matching operator]
An operator $O$ is in BKE-matching operator family $\mathcal{T}_{em}$, if and only if
\begin{eqnarray}
H_{\mathrm{CBC},\alpha}(O)+H_{\mathrm{BBC},\alpha}(O)= 2n-\left(H_{BKE,1/2}(O)+H_{FBKE,1/2}(O)\right).
\end{eqnarray}
\end{definition}
\noindent Examples of BKE-matching operators are $O=|i\rangle\langle j|$ or $O=P_i$. In fact, it is easy to find that as long as $O$ is a tensor product of single-qubit operators with each single-qubit operator either a computational basis or a Pauli basis, $O$ is a BKE-matching operator. For a BKE-matching operator, BKE (and Fourier-BKE) exactly tells us its CBC and BBC, and therefore strictly bounds SE
\begin{eqnarray}
G_{zero,\alpha}(O)\leq \min\left\{n-H_{FBKE,1/2}(O),n-H_{BKE,1/2}(O)\right\}-H_{SE,\alpha}(O),\quad O\in\mathcal{T}_{em}.
\end{eqnarray}

A direct implication of the BKE-matching operators is in the following corollary.
\begin{corollary}
A BKE-matching operator $O$ has minimal CBC-BBC uncertainty
\begin{eqnarray}
H_{\mathrm{CBC},\alpha}(O)+H_{\mathrm{BBC},\alpha}(O)=n,\quad O\in \mathcal{T}_{em}.
\end{eqnarray}
\end{corollary}
\begin{proof}
Since we have
\begin{eqnarray}
H_{\mathrm{CBC},\alpha}(O)+H_{\mathrm{BBC},\alpha}(O)= 2n-\br{H_{BKE,1/2}(O)+H_{FBKE,1/2}(O)}\geq n,
\end{eqnarray}
and 
\begin{eqnarray}
H_{BKE,1/2}(O)+H_{FBKE,1/2}(O)\geq H_{BKE,1/2}(O)+H_{FBKE,\infty}(O)\geq n,
\end{eqnarray}
where the first inequality is from Eq. \ref{renyiorder}, therefore, we can prove
\begin{eqnarray}
H_{BKE,1/2}(O)+H_{FBKE,1/2}(O)&&= n\\
H_{\mathrm{CBC},\alpha}(O)+H_{\mathrm{BBC},\alpha}(O)&&=n.
\end{eqnarray}
\end{proof}
\begin{lemma}\label{ll7}
    Given an BKE-matching operator $O$ with bra-ket entanglement $H_{BKE,1/2} ({O})= r \leq n $. If circuit $U$ composed of   $\{H,S,CNOT,T\}$ gates satisfied $H_{SE,\alpha}(U^\dagger OU) \geq \eta$, then the $T$ gate number within $U$ satisfied $N_T \geq \eta - (n-r)$ and $N_{H} \geq \eta - r$.
\end{lemma}
\begin{proof}
       Combined with the condition below for BKE-matching operators:
    \begin{eqnarray}
H_{BKE,1/2}(O)+H_{FBKE,1/2}(O)&&= n,\\
H_{\mathrm{CBC},\alpha}(O)+H_{\mathrm{BBC},\alpha}(O)&&=n,
\end{eqnarray}
We conclude that
    \begin{align}
     n - H_{BKE,1/2}(O)  &\leq H_{\mathrm{BBC},\alpha}(O)\\
     &= n -H_{\mathrm{CBC},\alpha}(O) \\
        &\leq H_{FBKE,1/2}(O)\\
        &= n - H_{BKE,1/2}(O),
    \end{align}
   which means the identity $H_{\mathrm{BBC},\alpha}(O) = n - H_{BKE,1/2}(O) = n-r$.
    Then noted that
\begin{align}
      H_{SE,\alpha}(U^\dagger OU)\leq  H_{\mathrm{BBC},\alpha}(U^\dagger OU) \leq H_{BBC,0}(U^\dagger O U)\leq N_T+H_{BBC,0}(O).
    \end{align}
 \begin{align}
        N_T &\geq H_{SE,\alpha}(U^\dagger OU)-H_{BBC,0}(O)\\
        &=H_{SE,\alpha}(U^\dagger OU)-(n - r)
    \end{align}
    Similarly, for the $H$ gates, we conclude that:
    \begin{align}
        H_{SE,\alpha}(U^\dagger OU)\leq  H_{\mathrm{CBC},\alpha}(U^\dagger OU) \leq H_{BBC,0}(U^\dagger O U)\leq 2N_H+H_{CBC,0}(O),
    \end{align}
    which leads to the result $N_{H} \geq (1/2)\cdot[H_{SE,\alpha}(U^\dagger OU)- r]$ based on the fact that $H_{\mathrm{BBC},\alpha}(O) =r$ for TE-matching operators.
\end{proof}
\color{black}
\section{$\mathcal{F}_{\mathrm{CBC}}$ vs $\mathcal{F}_{\mathrm{BBC}}$}
\subsection{Magic-coherence Duality\label{mcd}}
$\mathcal{F}_{\mathrm{CBC}}$ is composed of permutation and diagonal phase gates on the computational basis, and $\mathcal{F}_{\mathrm{BBC}}$ are Clifford circuit with also permutation and phase effects on the Pauli basis. However, $\mathcal{F}_{\mathrm{CBC}}$ is not merely Clifford circuits with the Pauli basis mapped to the computational basis. The reason is that $U$ can only transform $O$ through $UOU^\dag$, therefore, Clifford circuits acting on a Pauli operator can only introduce phases $\pm1$. In contrast, $U$ in $\mathcal{F}_{\mathrm{CBC}}$ can introduce complex phases for the computational basis. In this sense, $\mathcal{F}_{\mathrm{CBC}}$ is larger than $\mathcal{F}_{\mathrm{BBC}}$. In fact, we can show that the Clifford structure can be defined in $\mathcal{F}_{\mathrm{CBC}}$ and is a subset of $\mathcal{F}_{\mathrm{CBC}}$.
\begin{definition}[Generalized Clifford group on basis $\mathcal{B}$]
\label{def:GC_on_B}
    Given a set of phase-sensitive and complete matrix basis $\mathcal{B}=\{{B}_j\}_{j = 1}^{N} \cup \{0\}$ such that ${B}_j \cdot{B}_k \in \mathcal{B}, B_j^\dagger \in \mathcal{B}$ and the basis product map $  B_j \cdot B_k \to B_l$ can be calculated within $\text{poly}(n)$ bra-ket. 
    The unitary transformation such that  $U^\dagger B_iU = B_{\pi (i)}$ is call the generalized Clifford of $\mathcal{B}$, which we denoted as $\mathcal{C}(\mathcal{B})$.
  The unitary  $U \in \mathcal{C}(\mathcal{B})$ is a Pauli components in $\mathcal{C}(\mathcal{B})$ iff $U = U^\dagger$ and $U^\dagger B_iU = \alpha_i B_i$, which we denote $\mathcal{P}(\mathcal{B})$. 
\end{definition}
The generalized Pauli gates have the property that it can commute with generalized Clifford gates, at a cost of a generalized Pauli gate correction. This is proved by the lemma below.
\begin{lemma}[Peoperty of generalized Pauli operators]
\label{lemma:comm_Pauli}
Generalized Pauli operators satisfied the properties below:\\
(i). Take $P \in \mathcal{P}(\mathcal{B}), \text{Cl} \in \mathcal{C}(\mathcal{B})$, then $\text{Cl}^\dagger P\text{Cl} \in \mathcal{P}(\mathcal{B})$.\\
 (ii). Element $P\in \mathcal{P}(\mathcal{B})$ is either commute with all elements in $\mathcal{P}(\mathcal{B})$ , which we called \textit{decoherence like}, or on the other hand, anti-commute with at least half of the elements within $\mathcal{P}(\mathcal{B})$, which we called \textit{magic-like}.\\
 (iii). If $\mathcal{B}$  has finite elements, then element $\text{Cl}\in \mathcal{C}(\mathcal{B})$ can be finitely generated by some basic gate set $\mathcal{GC}(\mathcal{B})$.
\end{lemma}
\begin{proof}
    It is trivial to check that $\text{Cl}^\dagger P \text{Cl}$ is Hermitian.
    By definition,   $\text{Cl}^\dagger B_i\text{Cl} = \alpha_i B_{\pi (i)}$, which means that $\text{Cl}B_{i}\text{Cl}^{\dagger} = \alpha_{\pi^{-1}(i)}^{-1}B_{\pi^{-1}(i)}$. Then it can be shown that
    \begin{eqnarray}
        \text{Cl}^\dagger P^\dagger\text{Cl} B_j \text{Cl}^\dagger P\text{Cl} &&=  \text{Cl}^\dagger P^\dagger \alpha_{\pi^{-1}(j)}^{-1} B_{\pi^{-1}(j)}^\dagger P\text{Cl} \\
        &&=\text{Cl}^\dagger \beta_{\pi^{-1}(j)} \alpha_{\pi^{-1}(j)}^{-1} B_{\pi^{-1}(j)}^\dagger\text{Cl} \\
        &&= \beta_{\pi^{-1}(j)} \alpha_{\pi^{-1}(j)}^{-1}\alpha_{\pi^{-1}(j)} \gamma_{\pi^{-1}(j)} B_{j}^\dagger.
    \end{eqnarray}
       This proves the property (i). By definition, for $U_1,U_2 \in\mathcal{P}(\mathcal{B})$,
    \begin{eqnarray}
         &&(U_1U_2)^\dagger B_i(U_1U_2) = \alpha_i^{(1)}\alpha_i^{(2)}B_i =(U_2U_1)^\dagger B_i(U_2U_1),\\
         &&(U_2U_1)(U_1U_2)^\dagger B_i(U_1U_2)(U_2U_1)^\dagger= B_i.
    \end{eqnarray}
    By defining  $R=U_2U_1U_2^\dagger U_1^\dagger$, we noted that $[R,B_i] = 0$.
    Because $B_i \in \mathcal{B}$ forms a complete basis of matrix, this means that $[R ,M] = 0$ for arbitrary matrix $M \in \mathbb{C}^n \times \mathbb{C}^n$ .
    This concludes that $R = \alpha \mathbb{I}, \alpha \in \mathbb{C}$.
   Because $U_1,U_2$ as hermitian by definition, we conclude that $U_2U_1 = \pm U_1 U_2$. 
   This shows that $U_1,U_2 \in \mathcal{P}(\mathcal{B})$ are either commute or anti-commute.
    Given arbitrary element $P \in \mathcal{P}(\mathcal{B})$, find the center $C(P)$ of $P$, which is the maximum set of operators that that commute with $P$.
    If $C(P) = \mathcal{P}(\mathcal{B})$,  then $P$ commute with all the other elements
     If $C(P) \neq \mathcal{P}(\mathcal{B})$,  then there must exist $Q \in \mathcal{P}(\mathcal{B})$ such that $PQ = -QP$. 
     Then elements in $Q\cdot C(P):= \{Q\cdot e| e\in C(P)\}$ all anti-commute with $P$. 
     Finally, we note that $|C(P)| = |Q\cdot C(P)|$ and $C(P) \cup Q\cdot C(P) = \mathcal{P}(\mathcal{B})$ because  $C(P)$ is maximum. 
     This proves the second statement.
Noted that $\mathcal{B}$ has finite elements, thus $ \mathcal{C}(\mathcal{B})$ is a subgroup of $\pi(|\mathcal{B}|)$, being another finite group.
Then Element $\text{Cl}\in \mathcal{C}(\mathcal{B})$ can be finitely generated by some basic gate set $\mathcal{GC}(\mathcal{B})$.
\end{proof}
A well known example is the Pauli basis $\mathcal{B}_{\text{Pauli}}:=\{\pm i,\pm1\} \times \{I,Z,X,Y\}^{\otimes n} \cup \{0\}$ and its Clifford gates $\mathcal{C}(\mathcal{B}_{\text{Pauli}}) $ generated by gates  $\{CX,S,H\} $. 
Other interesting examples are (i).unitary equivalent Pauli basis $\mathcal{B}_{\text{Pauli}}(R):= \{R PR^\dagger | P \in \mathcal{B}_{\text{Pauli}}\}$, (ii).stabilizer basis  $\mathcal{B}_{\text{STAB}}:=\{2^{-p/2}\omega^m|\omega = e^{i\pi/4},p \in \mathbb{Z}_{\geq 0},m \in \mathbb{Z}_8\} \times  \{\ket{\psi}\!\bra{\phi} | \ket{\psi},\ket{\phi} \in \text{STAB}_n\}$, (iii).mutually unbiased basis $\mathcal{B}_{\text{MUB}}:=\{z \in \mathbb{C}||z| = d^{-p/2}, p  = \mathbb{Z}_{\geq 0}\} \times  \{\ket{\psi}\!\bra{\phi} | \ket{\psi},\ket{\phi} \in \text{MUB}(d)\}$, and (iv).Pauli-computational mixed basis $\mathcal{B}_{\text{MIX}}:=\{\pm i,\pm1\} \times \{I,Z,X,Y\}^{\otimes r} \otimes \{\ket{\mathbf j}\!\bra{\mathbf k}\}_{\mathbf j,\mathbf k \in \{0,1\}^{ n -r}}\cup \{0\}$.

The choice related to coherence, which we will further discussed in this section, is the phase-restrained computational basis 
$\mathcal{B}_{\text{comp}} =  \left\{ \{\pm i, \pm 1\} \times\ket{\mathbf j}\!\bra{\mathbf k}\right\}_{\mathbf j,\mathbf k \in \{0,1\}^{ n}}\cup \{0\}$ with $4 \times 4^n + 1$ elements.
For $\mathcal{B} = \mathcal{B}_{\text{comp}}$, The set $\mathcal{C}(\mathcal{B})$ and $\mathcal{P}(\mathcal{B} )$ forms a group and actually has the form
 \begin{align}
    &\mathcal{C}(\mathcal{B})\!=\!\left\{\!\!\sum_{j \in \{0,1\}^n}\!\!\!\!\!e^{i\frac{\pi}{2}\phi(j) } \!\ket{\pi(j)}\!\bra{j} ;\phi(j)\!\!:\!\!\{0,1\}^n \to \mathbb{Z}_4, \!\pi \!\in\! S_n\!\right\} ,\label{eq:comp_C}\\
    &\mathcal{P}(\mathcal{B})\! =\! \left\{\!\!\sum_{j \in \{0,1\}^n}e^{i\frac{\pi}{2}\phi(j) } \ket{j}\!\bra{j} ;\phi(j): \{0,1\}^n \to \{0,2\}\right\}.\label{eq:comp_P}
\end{align}
It is also simple to see that both $\mathcal{C}(\mathcal{B}_{\text{comp}})$ and $\mathcal{P}(\mathcal{B}_{\text{comp}})$ forms group and inherits many properties in ordinary Clifford gates, as proved below.
\begin{lemma}[Property of  $\mathcal{B}_{\text{comp}}$]
\label{lemma:Generalized_Cliff_comp}
    The group $\mathcal{C}(\mathcal{B}_{\text{comp}} )$ has the form Eq.~\eqref{eq:comp_C} and can be finitely generated by gate set $\mathcal{GC}(\mathcal{B}_{\text{comp}})=\{CX,S,CCX\}$.
            The group $\mathcal{P}(\mathcal{B}_{\text{comp}} )$ has the form Eq.~\eqref{eq:comp_P} and composed of real-valued diagonal gates that encodes arbitrary Boolean function $f: \{0,1\}^n \to \{+1,-1\}$.
\end{lemma}
\begin{proof}
    By definition, for $U \in \mathcal{C}(\mathcal{B}_{\text{comp}}) $, $U \ket{j}\!\bra{k}U^\dagger = \alpha_{jk} \ket{\pi_L(j,k)}\!\bra{\pi_R(j,k)}$, where $\pi_L,\pi_R $ are permutation of length-$2n$ binary strings. The phase shift satisfied $\alpha_{ij} =e^{i\frac{\pi}{2}\phi(i,j) }$ with $\phi(:,:)$ a map $\{0,1\}^n \times \{0,1\}^n \to \mathbb{Z}_4$ because the constant before the basis are in the set $\{\pm 1,\pm i\}$.
    Similarly, we see that $U \ket{k}\!\bra{j'}U^\dagger = \alpha_{kj'} \ket{\pi_L(k,j')}\!\bra{\pi_R(k,j')}$ . 
    Taking $j = j'$, then conjugation gives $\alpha_{jk} = \alpha_{kj}^*$.
    Multiplying these two equalities will gives
    \begin{eqnarray}
        U \ket{j}\!\bra{j'}U^\dagger&&= \alpha_{jk}\alpha_{kj'} \ket{\pi_L(j,k)}\!\bra{\pi_R(j,k)}{\pi_L(k,j')}\rangle\!\bra{\pi_R(k,j')}\\
        &&= \alpha_{jk}\alpha_{kj'} \delta_{{\pi_R(j,k)},{\pi_L(k,j')}} \ket{\pi_L(j,k)}\!\bra{\pi_R(k,j')}\\
        &&=\alpha_{jj'} \ket{\pi_L(j,j')}\!\bra{\pi_R(j,j')}.
    \end{eqnarray}
  Because $U$ is unitary, we conclude that $\alpha_{jj} \neq 0$. Then for arbitrary $k,j \in \{0,1\}^n$ , we have
    \begin{eqnarray}
\label{eq:perm_ident}
\alpha_{jk}\alpha_{kj} &= \alpha_{jj} = 1,\pi_R(j,k) &= \pi_L(k,j) ,\\
            \pi_L(j,j) &= \pi_L(j,k),  \pi_R(k,j) &= \pi_R(j,j).
    \end{eqnarray}
 It is thus convenient to denote $\pi_L(j):=\pi_L(j,k) \equiv \pi_L(j,j)$ and $\pi_R(j) := \pi_R(k,j)\equiv \pi_R(j,j)$ .
 Because the second line in Eq.~\eqref{eq:perm_ident}, we conclude $\pi_L(k) = \pi_R(k):=\pi(k)$. 
 Then we conclude
 \begin{eqnarray}
     U \ket{j}\!\bra{k}U^\dagger U \ket{k}\!\bra{j'}U^\dagger &= \alpha_{jk} \alpha_{kj'} \ket{\pi(j)}\!\bra{\pi(j')} =\alpha_{jj'} \ket{\pi(j)}\!\bra{\pi(j')},
 \end{eqnarray}
 which means that $\phi(j,k) - \phi(j',k) = \phi(j,j'), \forall j,j',k \in \{0,1\}^n$, fixing the phase function as the form $\phi(j,k) := \phi(j) - \phi(k)$.
 Based on analysis above, valid unitary operations must have form
 \begin{eqnarray}
     U \ket{j}\!\bra{k}U^\dagger &&= \alpha_{jk} \ket{\pi(j)}\!\bra{\pi(k)}\\
    && =\left(\sum_{j' \in \{0,1\}^n}e^{i\frac{\pi}{2}\phi(j') } \ket{\pi(j')}\!\bra{j'} \right) \ket{j}\!\bra{k} \left(\sum_{k' \in \{0,1\}^n}e^{-i\frac{\pi}{2}\phi(k') } \ket{k'}\bra{\pi(k')}\right)
 \end{eqnarray}
 This prove our first claim.

 To show that  $\{CCX,S,X\}$ forms a complete generator set of this group, it is trivial to see that $CX_{12}CX_{21}CX_{12}= SWAP_{12}$ forms permutation between any two bits, which can generate arbitrary permutation in $S_n$ group.
 Thus the problem reduced to construct phase operation $D\ket{j} \to e^{i \frac{\pi}{2} \phi(j)} \ket{j}$ for $D \in \mathcal{GC}(\mathcal{B}_{\text{comp}})$.
 Noted that every function $\phi:\{0,1\}^n \to \mathbb{Z}_4$ can be uniquely expressed as a multilinear polynomial
 \begin{eqnarray}
 \phi(j) = \sum_{S \subseteq[n]} a_Sj^S    ,\text{where~}a_S \in \mathbb{Z}_4, j^S:= \prod_{t \in S}j_t.
 \end{eqnarray}
 
Combined with arbitrary $SWAP$ gate, it is sufficient to construct gate $D(a,t)\ket{j} \to e^{i \frac{a\pi}{2}j_1\cdots j_t} \ket{j}$ for arbitrary $a \in \mathbb{Z}_4$ and $t \in [n]$. Notice that If we can construct the operation 
\begin{eqnarray}
    M(t) \ket{j_1,\cdots,j_n} = M(t)\left| j_1,\cdots,j_{t-1},\prod_{r =1}^tj_r,j_{t+1},\cdots,j_n\right\rangle,
\end{eqnarray}
We can use de-computation to construct gate $D(a,t)$ as:
\begin{eqnarray}
 M(t)^{-1}  S_{j_t}^{a}  M(t) \ket{j_1,\cdots,j_n}  = e^{i \frac{a\pi}{2}j_1\cdots j_t} \ket{j_1,\cdots,j_n}.
\end{eqnarray}
Using the fact that $\{CCX,X\}$ constitute a universal reversible gate set in classical Boolean circuit~\cite{fredkin1982conservative}, we conclude the existence of $M(t)$.
    The next statement is explicit from Eq.~\eqref{eq:comp_C}, with the hermitian property.
\end{proof}

\subsection{$CCX$ vs $T$\label{whyccx}}
\begin{figure}
\centering
\includegraphics[width=0.7\linewidth]{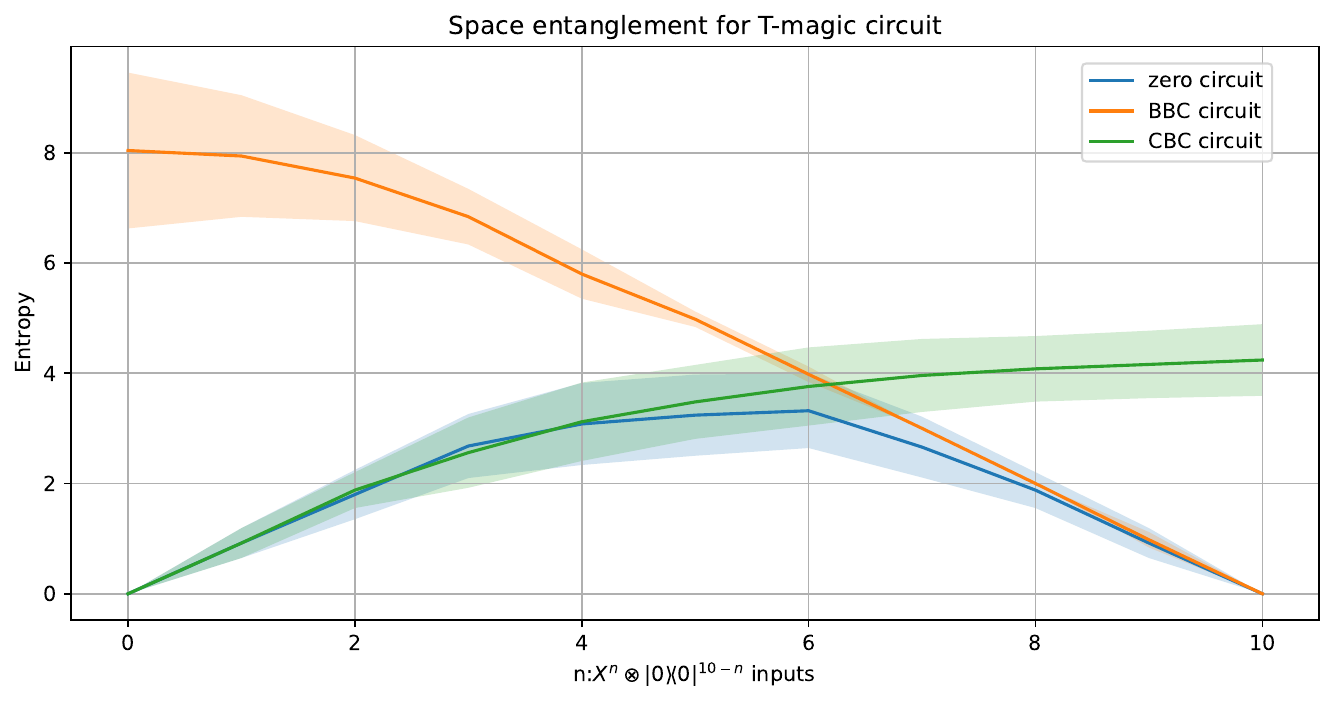}
\caption{The same setting as Fig \ref{fig3} except for $U$ in $\mathcal{F}_{\mathrm{CBC}}$, we construct them randomly from $\{T,S,CX\}$ instead of $\{CCX,T,S,CX\}$.}
\label{fig:T_and_Toffoli}
\end{figure}
In Fig. \ref{fig:T_and_Toffoli}, we numerically find that when $U$ in $\mathcal{F}_{\mathrm{CBC}}$ is made up of $\{T,S,CX\}$, SE cannot approach the upper bound. The simulations were performed using mindquantum package \cite{xu2024mindspore,ChanceSiyuan_repo}.
The reason is that $\{T,S,CX\}$ can only cover a subset of $\mathcal{F}_{\mathrm{CBC}}$. For example, $\{T,S,CX\}$ cannot construct $CCX$ \cite{nielsen2010quantum}. Therefore, we also introduce $CCX$ into the gate set and find that $CCX$ alone without $T$ can already approach the upper bound well as shown in Fig. \ref{fig3}. The reason can be understood from Section \ref{mcd}, where we show $\{CCX,CX,S\}$ forms the Clifford structure on computational basis, making a pretty good exploration on $\mathcal{F}_{\mathrm{CBC}}$ \cite{webb2015clifford}.

\section{Application of BKE-matching operators\label{app:application}}
We identify three physically motivated settings where 
BKE-matching operators with intermediate BKE arise 
naturally, demonstrating that the regime beyond zero 
BKE is essential for modern quantum simulation tasks.

\paragraph{Marginal probability distributions.}

Weak simulation of quantum circuits can be reduced to 
computing marginal probability 
distributions~\cite{bravyi2022simulate}. For an 
$n=(n_1+n_2)$-qubit system, the marginal over $n_2$ 
qubits reads
\begin{align}
p(\mathbf{s}) = \mathrm{tr}\!\left[U |0\rangle\!\langle 0|^{\otimes n} U^\dagger 
\cdot \left(\mathbb{I}_{n_1} \otimes |\mathbf{s}\rangle\!\langle\mathbf{s}|\right)\right].
\end{align}
The observable 
$O = \mathbb{I}_{n_1}\otimes|\mathbf{s}\rangle\!\langle\mathbf{s}|$ 
is a tensor product of single-qubit identity and 
computational-basis operators, hence a BKE-matching 
operator with 
$H_{\mathrm{BKE},1/2}(O) = n_1$ and 
$H_{\mathrm{FBKE},1/2}(O) = n_2$, giving 
$H_{\mathrm{CBC},\alpha}(O) = n_2$ and 
$H_{\mathrm{BBC},\alpha}(O) = n_1$.
Applying Lemma~\ref{ll7}, if a Clifford+$T$ circuit $U$ 
achieves $H_{\mathrm{SE},\alpha}(U^\dagger OU)\geq\eta$, 
then the gate counts satisfy
\begin{align}
N_T \geq \eta - n_2, \qquad N_H \geq \eta - n_1.
\end{align}
For the extremal cases, $n_1 = 0$ (amplitude estimation) 
recovers the purely coherence-limited regime with cost 
scaling as $2^{N_H}$, while $n_1 = n$ (full distribution) 
recovers the magic-limited regime with cost $2^{N_T}$. 
At intermediate values $n_1 \approx n_2 \approx n/2$, both 
bounds impose comparable constraints, confirming that 
neither computational-basis nor Pauli-basis methods alone 
are efficient. 

\paragraph{Mixed-state inputs. }
Experimental quantum devices inevitably produce mixed 
initial states due to imperfect preparation. Consider 
$n_1$ depolarized qubits subject to a Pauli error $P$ 
and $n_2$ well-initialized qubits:
\begin{align}
\rho_0 = \frac{\mathbb{I}_{n_1}+\varepsilon P}{2^{n_1}}
\otimes |0\rangle\!\langle 0|^{\otimes n_2}
= 2^{-n_1} O_1 + 2^{-n_1}\varepsilon\, O_2,
\end{align}
where $O_1 = \mathbb{I}_{n_1}\otimes|0\rangle\!\langle 0|^{\otimes n_2}$ 
and $O_2 = P\otimes|0\rangle\!\langle 0|^{\otimes n_2}$. 
Both $O_1$ and $O_2$ are BKE-matching operators, being 
tensor products of single-qubit Pauli and 
computational-basis operators. This decomposition is 
exact and requires only two terms, whereas a 
computational-basis decomposition of $\mathbb{I}_{n_1}$ 
would require $2^{n_1}$ terms and a full Pauli 
decomposition of $|0\rangle\!\langle 0|^{\otimes n_2}$ 
would require $2^{n_2}$ terms. The output probability 
$p(\mathbf{x}) = \mathrm{tr}(U\rho_0 U^\dagger \cdot 
|\mathbf{x}\rangle\!\langle\mathbf{x}|)$ thus reduces to 
two BKE-matching calculations, each analyzable via 
Lemma~\ref{ll7} with the BKE determining the 
coherence-versus-magic resource balance.

\paragraph{Quantum error correction protocols.}
Syndrome extraction in fault-tolerant quantum computing 
combines both of the above settings. Data qubits 
($n_1$ qubits, subject to logical Pauli errors) and 
ancilla qubits ($n_2$ qubits, initialized to 
$|0\rangle\!\langle 0|^{\otimes n_2}$) undergo a syndrome 
extraction circuit $U$. The syndrome probability is
\begin{align}
p(\mathbf{x},\mathbf{s}) = \mathrm{tr}\!\left[U\rho_0 U^\dagger 
\cdot \left(\mathbb{I}_{n_1}\otimes|\mathbf{s}\rangle\!\langle\mathbf{s}|\right)\right],
\end{align}
where $\rho_0 = 2^{-n_1}O_1 + 2^{-n_1}\varepsilon\,O_2$ 
as before. This calculation involves the contraction of 
BKE-matching operators from both the input 
($O_1$, $O_2$) and the output measurement 
($\mathbb{I}_{n_1}\otimes|\mathbf{s}\rangle\!\langle\mathbf{s}|$), 
all residing in the intermediate BKE regime. 
The simulation cost for each term is quantified by 
Lemma~\ref{ll7}, with the BKE values of input and output 
operators jointly determining whether coherence or magic 
resources dominate.
\color{black}
\end{appendix}
\end{document}